\documentclass[10pt]{llncs}
\usepackage{amssymb,amsmath,verbatim,booktabs}
\usepackage[sort]{cite}
\usepackage{xcolor}
\usepackage{ntheorem}
\usepackage{thmtools}
\usepackage{mathtools}
\usepackage{enumitem}
\usepackage{graphicx}
\usepackage[margin=1.2in]{geometry}
\usepackage{url}
\usepackage[breaklinks=true]{hyperref} 
\colorlet{green}{green!50!black}
\hypersetup{
	colorlinks=true,
	citecolor=green,
	linkcolor=blue,  
	urlcolor=black,
}
\usepackage{cleveref}
\usepackage{algorithm}
\usepackage{algorithmicx}
\usepackage[noend]{algpseudocode}
\usepackage{pifont}

\usepackage[title]{appendix}
\makeatletter
\newcommand{\@chapapp}{\relax}%
\makeatother

\def\checkmark{\ding{52}}
\def\xmark{\ding{55}}
\def\halfcheckmark{\ding{52}$^{\text{\bf \O}}$}

\def\Qi{\hyperref[quest:q1]{Q1}}
\def\Qii{\hyperref[quest:q2]{Q2}}
\def\Qiii{\hyperref[quest:q3]{Q3}}
\def\Qiv{\hyperref[quest:q4]{Q4}}
\def\Qv{\hyperref[quest:q5]{Q5}}
\def\Qvi{\hyperref[quest:q6]{Q6}}

\declaretheorem[name=Claim,numbered=no]{cclaim}

\newcommand{\arcs}{\textsc{Arc Restrictions for Strict Core}}
\newcommand{\ec}{\textsc{Arc in Core}}
\newcommand{\fec}{\textsc{Forbidden Arc in Core}}

\newcommand{\SR}{\textsc{Stable Roommates}}
\newcommand{\acycpart}{\textsc{Acyclic Partition}}
\newcommand{\algoSR}{\textup{SR-Improve}}
\newcommand{\algoHM}{\textup{HM-Improve}}
\newcommand{\GG}{\widetilde{G}}
\newcommand{\qq}{\widetilde{q}}
\newcommand{\QQ}{\widetilde{Q}}
\newcommand{\NN}{\widetilde{N}}
\newcommand{\HH}{\widetilde{H}}
\newcommand{\EE}{\widetilde{E}}
\newcommand{\OPT}{\textup{OPT}}
\newcommand{\leteq}{\vcentcolon =}

\begin{document}

\author{Ildik\'o Schlotter\inst{1,2} \and P\'eter Bir\'o\inst{1,3} \and Tam\'as Fleiner\inst{1,2}}

\title{
The core of housing markets from an agent's perspective:\\ Is it worth sprucing up your home?
}
\institute{
Centre for Economic and Regional Studies, Budapest, Hungary, \\
\email{\{schlotter.ildiko,biro.peter\}@krtk.hun-ren.hu}
\and
Budapest University of Technology and Economics, Budapest,  Hungary \\
\email{fleiner.tamas@vik.bme.hu}
\and
Corvinus University of Budapest, Budapest, Hungary}

\maketitle

\begin{abstract}
We study housing markets as introduced by Shapley and Scarf (1974).
We investigate the computational complexity of various questions regarding the situation of an agent~$a$ in a housing market~$H$:
we show that it is $\mathsf{NP}$-hard to find an allocation in the core of $H$ where (i) $a$ receives a certain house, 
(ii) $a$ does not receive a certain house, or (iii) $a$ receives a house other than her own.
We prove that the core of housing markets \emph{respects improvement} in the following sense: 
given an allocation in the core of~$H$ where agent~$a$ receives a house~$h$, if the value of the house owned by~$a$ increases, then 
the resulting housing market admits an allocation in its core in which $a$ receives either~$h$ or a house that $a$ prefers to~$h$;
moreover, such an allocation can be found efficiently.
We further show an analogous result in the \textsc{Stable Roommates} setting by proving that stable matchings in a one-sided market also respect improvement.  
\end{abstract}


\section{Introduction.}
\label{sec:intro}

Housing markets are a classic model in economics where agents are initially endowed with one unit of an indivisible good, called a \emph{house}, 
and agents may trade their houses according to their preferences without using monetary transfers. 
In such markets, trading results in a reallocation of houses in a way that each agent ends up with exactly one house.
Motivation for studying housing markets comes from applications such as
kidney exchange~\cite{roth-sonmez-unver-2004,biro-kidney-exchange-survey,biroetal2021} and on-campus housing~\cite{adbulkadiroglu-sonmez}.

In their seminal work Shapley and Scarf~\cite{shapley-scarf-1974} examined housing markets where agents' preferences are weak orders. 
They proved that such markets always admit a \emph{core} allocation, that is, an allocation where no coalition of agents 
can strictly improve their situation by trading only among themselves.
They also described the Top Trading Cycles (TTC) algorithm, proposed by David Gale,
and proved that the set of allocations that can be obtained through the TTC algorithm 
coincides with the set of competitive allocations; hence the TTC always produces an allocation in the core. 
When preferences are strict, the TTC produces the unique allocation in the \emph{strict core}, 
that is, an allocation where no coalition of agents can weakly improve their situation (with at least one agent strictly improving) by trading among themselves~\cite{roth-postlewaite}. 
Although the strict core has some very appealing properties from a mathematical viewpoint (above all, that for strict preferences it contains a unique allocation that is easy to compute), there are arguments why the core is a more interesting solution concept. 
First, if we allow indifference between houses to appear in the preferences, then the strict core can be empty.
Second, decision-makers in a real-world application have to deal with various constraints and optimization goals that may not be represented in the preferences. 
In kidney exchange programs such constraints can arise due to ethical issues~\cite{Chow-KE-ethics}, logistical considerations (e.g., limit on the length of the exchange cycles),
and several optimization criteria that can improve the accessibility and long term success of the programs (e.g., prioritization of highly sensitized patients, or keeping donors with blood type O for recipients with blood type O in order to avoid the accumulation of hard-to-match patients); see \cite{biroetal2021} for a survey on European practices.
Thus, having a wider range of possible solutions to choose from is often  preferable.

Bir\'o et al.~\cite{BKKV-mor} conducted computer simulations on realistic kidney exchange instances to compare solutions of maximum size or maximum weight\footnote{In practice, 
solutions in a kidney exchange program are often sought as maximum-weight matchings between patient--donor pairs in a graph where weights reflect certain optimality criteria.},
 typically used in practice, with solutions contained in the core, in the set of competitive solutions, and in the strict core of the underlying housing market.
 Intuitively, as the latter three solution concepts become increasingly demanding to be satisfied, the ``price of fairness'' (roughly speaking, the reduction in the number of transplants due to taking preferences into account) also increases. 
Their results suggest that the core is a good compromise: they found that ``core allocations for instances with 150 patient--donor pairs entail a less than 1\% reduction in the number of transplants.'' 
In contrast, regarding the strict core solutions, they observed that the likelihood of existence for instances with weak preferences is decreasing sharply as the kidney exchange pool grows, becoming almost zero for 150 patient--donor pairs.

\medskip

Although the core of housing markets has been the subject of considerable research, there are still many challenges which have not been addressed.
Our paper focuses on questions that may be raised by an agent who wants to decide whether to enter the market, and if so, on what conditions. To be able to judge their prospects correctly, agents are in need of  information about the possible allocations the market may result in. Can they improve their situation by participating in the market? If so, how much? Which are the houses that they have a chance of receiving?

Assuming that our relevant solution concept is the core, an agent~$a$ may be interested in the following questions:
\begin{enumerate}[label=$\bullet$]
    \item Q1: \label{quest:q1} Can agent~$a$ receive a house better than her own in some core allocation?
    \item Q2: \label{quest:q2} Given some house~$h$, can agent~$a$ obtain $h$ in some core allocation?
    \item Q3: \label{quest:q3} Given some house~$h$, can agent~$a$ avoid obtaining~$h$ in some core allocation?
\end{enumerate}

We believe that the above three questions are natural enough to warrant a quest for an efficient algorithm that can solve the underlying computational problems. However, their motivation is also clear from an economic point of view: 
a positive answer to question \Qi{} (or \Qii{}) clearly provides a strong incentive for agent~$a$ to participate in the market. Hence, such an algorithm can be an important tool for an authority in charge of a centralized housing market where a larger market is more desirable (as is the case, e.g., in kidney exchange programs).
Further motivation for studying questions \Qii{} and \Qiii{} stems from the fact that, in some cases, realizing a given allocation requires certain 
investments that are a function of the allocation chosen: e.g., in the context of kidney exchange, additional compatibility tests are required before carrying out the planned transplantations.
Narrowing down the set of possible donors whose kidney a given patient may obtain in a kidney exchange may allow for such tests to be performed in advance, sparing time for the patients, while keeping the costs incurred by such tests relatively low.\footnote{As an extreme example of additional expenditures necessary for realizing an allocation, consider Germany, where kidney exchange is only allowed when it involves people who have a ``personal relationship.'' However, the required personal relationship can be built for the purpose of the kidney exchange to take place---a process that takes precious time for the patients involved. See \url{https://crossover-nierenspenderliste.de/files/DIATRA_42021_26-29_Crossover_engl.pdf}.} 

In the first part of our work, we focus on the computational complexity of the above questions.
Similar questions have been extensively studied in the context of the \textsc{Stable Marriage}
and the \textsc{Stable Roommates} problems~\cite{Knuth1976,GusfieldIrving-book,Dias-2003,Fleiner-Irving-Manlove,Cseh-Manlove-SR-2016}, 
but have not yet been considered in relation to housing markets.

We also address questions that concern the possibility of an agent improving her situation  by bringing a better endowment to the market. 
Assuming that agent~$a$ ensures that the value of her house increases, will this result in an improvement for~$a$? If the answer is positive, then this provides an incentive for the agent to invest in her house in order to obtain a preferable allocation. 
It is clear that an increase in the value of $a$'s house may not always yield a \emph{strict} improvement for~$a$ 
(as a trivial example, some core allocation may assign~$a$ her top choice even before the change), 
but is it at least true that by improving her house, $a$ will not damage her own possibilities in the market?
Can we determine whether when a strict improvement for~$a$ becomes possible?

We investigate the following question: is an increase in the value of some agent~$a$'s house beneficial for~$a$ in terms of the possible core allocations? 
More precisely, we consider two slightly different versions of this question: 
\begin{enumerate}[label=$\bullet$]
    \item \label{quest:q4} Q4: Given a core allocation for the original market where $a$ obtains some house~$h$, 
    can agent~$a$ obtain a house \emph{at least as good} as~$h$ in some core allocation after an increase in the value of $a$'s house?
    \item \label{quest:q5} Q5: Given a core allocation for the original market where $a$ obtains some house~$h$, 
    can an agent~$a$ obtain a house \emph{strictly better} than~$h$  in some core allocation after an increase in the value of~$a$'s house?
\end{enumerate}

Questions \Qiv{} and \Qv{} are of crucial importance when we consider agents' incentives to choose the endowment with which they enter the market. In the context of kidney exchange: if procuring a new donor with better properties 
(e.g., a younger or healthier donor),
or registering an additional willing donor (which is possible in most European programs) 
does not necessarily benefit the patient, then this could undermine the incentive for the patient 
to find a donor with good characteristics, damaging the overall welfare in a system where inefficiency directly leads to loss of lives. 
Being able to answer the above questions is therefore paramount, and has direct consequences on agents' incentives.

\subsection{Our contribution.}
Regarding questions \Qi{}, \Qii{}, and \Qiii{} raised above, we show in Theorem~\ref{thm:arc-in-core} that each of them is computationally intractable. 
Remarkably, it is already $\mathsf{NP}$-complete to decide whether a core allocation can assign \emph{any} house to~$a$ other than her own. 
Similarly, deciding whether the core of a housing market contains 
an allocation where a given agent~$a$ obtains a certain house (or where $a$ does \emph{not} receive a certain house) is also $\mathsf{NP}$-complete.
Various generalizations of these questions can be answered efficiently in both the \textsc{Stable Marriage} and \textsc{Stable Roommates}
settings~\cite{Knuth1976,GusfieldIrving-book,Dias-2003,Fleiner-Irving-Manlove,Cseh-Manlove-SR-2016}, 
so we find these intractability results surprising.

We shall note that our complexity results do not mean that finding core allocations that fulfill some additional requirement would be impossible in practice, it only means that a polynomial-time algorithm is  highly unlikely to exist for such problems. However, this does not preclude the use of 
heuristics or robust optimization techniques such as integer programming (IP) methods for computing core allocations with additional requirements. In fact, Bir\'o et al.~\cite{BKKV-mor} developed and tested such methods by conducting simulations, and demonstrated the possibility of using the solution concept of core in practice.

\smallskip
Turning our attention to the question of how an increase in the value of a house affects its owner, 
we present Theorem~\ref{thm:core-main-RI}, our main technical result, which answers question~\Qiv{} affirmatively as follows:

\begin{center}
\begin{minipage}{15cm}%
\emph{If the core of a housing market contains an allocation where $a$ receives some house $h$, 
and the market changes in a way that some agents perceive an increased value for the house owned by $a$ 
(and nothing else changes in the market), then the resulting housing market admits an allocation in its core where 
$a$ receives either~$h$ or a house that $a$ prefers to~$h$.}
\end{minipage}
\end{center}

Using the terminology of Bir\'o et al.~\cite{BKKV-mor}, the above result shows that the core \emph{respects improvement} in the sense that the best allocation achievable for an agent~$a$ in a core allocation can only (weakly) improve for $a$ as a result of an increase in the value of $a$'s house.
We prove Theorem~\ref{thm:core-main-RI} by presenting a polynomial-time algorithm that finds an allocation as promised by the sentence highlighted above; the ideas and techniques on which this algorithm relies is our main technical contribution.
This settles an open question asked explicitly by Bir\'o et al.~\cite{BKKV-mor}.

This result has important implications for the practice of kidney exchanges. Bir\'o et al.~\cite{BKKV-mor} conducted simulations for measuring how often the property of respecting improvement is violated when using solutions of  maximum size/weight  (as done in practice) compared to using core, competitive, or strict core solutions. 
They observed a significant number of violations for 
 solutions of maximum size/weight, but none for core solutions, conjecturing the theorem that we have proved theoretically.
Therefore, Theorem~\ref{thm:core-main-RI} gives a theoretical foundation for their observation, and it implies that the usage of core solutions provides good individual incentives for recipients to bring better or more donors. In the meantime, these simulations show that the current practice of focusing only on the number of transplants and other weighted optimality criteria may not provide a compelling incentive for participants to bring valuable donors to the pool.

The significance of our positive result in Theorem~\ref{thm:core-main-RI} is especially pronounced
in view of the intractability results of Theorem~\ref{thm:arc-in-core}:
even though we cannot efficiently compute information about the possible houses an agent may obtain in some core allocation, we \emph{do know} that entering the market with a more valuable house, an agent can only improve (and never damage)  their situation. Hence, improving the value of the house with which the agent plans to enter the market is always a safe choice. 
We believe that this aspect of the core is a very strong argument for considering it as a good solution concept to be used in centralized housing markets, since it provides a motivation for agents to increase the value of their initial endowment.

\smallskip
Contrasting our positive result for question~\Qiv{}, the slightly different question~\Qv{} turns out to be significantly harder:
although one can formulate several variants of this problem depending on what exactly one considers to be a strict improvement, 
by Theorem~\ref{thm:strict-imp}
each of them leads to computational intractability
($\mathsf{NP}$-hardness or $\mathsf{coNP}$-hardness).

\smallskip
Finally, we also answer a question raised by Bir\'o et al.~\cite{BKKV-mor} regarding the property of respecting improvements 
in the context of the \SR{} problem. An instance of \SR{} contains a set of agents, each having preferences over the other agents; 
the usual task is to find a matching between the agents that is \emph{stable}, i.e., no two agents prefer each other to their partners in the matching. 
An instance of \SR{} can therefore be considered as a housing market with the additional requirement that (i) trading can only happen along cycles of length~2, and (ii) only blocking cycles of length~2 can cause instability; then stable matchings correspond exactly to core allocations. We examine the following question, which is the direct analog of question~\Qiv{} for the \SR{} model:  
\begin{enumerate}[label=$\bullet$]
    \item \label{quest:q6} Q6: In an instance of \SR{}, does increasing the value of an agent~$a$ (as manifested in the preferences of others) lead to a (weak) improvement in the situation of~$a$?
\end{enumerate}

Again, we are able to assert a positive answer, although only in a conditional form:
in Theorem~\ref{thm:SR-RI-holds} we show that 
if some stable matching assigns agent~$a$ to agent~$b$ in a \SR{} instance, and the 
value of~$a$ increases (that is, if $a$ moves upward in other agents' preferences, with everything else remaining constant), 
and the resulting instance admits a stable matching, 
then it necessarily admits a stable matching where $a$ is matched either to~$b$ or to an agent preferred by~$a$ to~$b$. 
This result is an analog of the one stated in Theorem~\ref{thm:core-main-RI} for the core of housing markets; 
however, the algorithm we propose in order to prove it uses different techniques. 

\smallskip
We remark that throughout the paper we use a model with partially ordered preferences (a generalization of weak orders). 
Although partially ordered preferences have been studied in the context of the \textsc{Stable Marriage} and \SR{} problems~\cite{Fleiner-Irving-Manlove,GPRVW2011,Drummon-Boutilier-2013, Pittel2020,Cseh-Juhos-2021}, we are not aware of any paper on housing markets featuring preferences that are expressed as partial orders.

\subsection{Related work.}
Most works relating to the core of housing markets
aim to find core allocations with some additional property that benefits global welfare,
most prominently Pareto optimality~\cite{Jaramillo-Manjunath,Alcalde-Unzu-Mollis,Aziz-deKeijzer,Plaxton-2013,Saban-Sethuraman-2013}.
Another line of research 
comes from kidney exchange where the length of trading cycles is of great importance and often plays a role in agents' preferences~\cite{cechlarova-romero-medina,cechlarova-hajdukova-2003,cechlarova-fleiner-manlove-kidney,Biro-Cechlarova-2007,cechlarove-lacko-2012}
or is bounded by some constant~\cite{abraham-blum-sandholm,biro-manlove-rizzi-kidney,biro-mcdermid-2010-3cycles,huang-3wayKE,Cechlarova-Repisky-2011}.
None of these papers deal with problems where a core allocation is required to fulfill 
some constraint regarding a given agent or set of agents---that they be trading, or that they obtain (or not obtain) a certain house. 
Although, to the best of our knowledge none of the questions \Qi{}--\Qiii{} have been studied so far, some papers have focused on finding a core allocation where the number of agents involved in trading is 
as large as possible, obtaining mostly intractability results~\cite{Biro-Cechlarova-2007,Cechlarova-Repisky-2011}.

In the context of the \textsc{Stable Marriage} and the \textsc{Stable Roommates} problems, it is known that the problem of finding a stable matching with edge restrictions, i.e., a stable matching that contains a given set of \emph{forced} edges but is disjoint from a given set of \emph{forbidden} edges, can be found in polynomial time~\cite{Dias-2003,Fleiner-Irving-Manlove}. These results strongly contrast Theorem~\ref{thm:arc-in-core}, which shows that the analogous problems in the context of house allocation are $\mathsf{NP}$-hard, even if there is only a single arc that we require to be included in (or excluded from) the desired allocation.

Questions \Qiv{} and \Qv{} can be considered as inquiries about housing markets where preferences are subject to change. Although some researchers have addressed certain dynamic models, most of these either focus on the possibility of repeated allocation~\cite{roth-postlewaite,kamijo-kawasaki,kawasaki-2015}, or consider a situation where agents may enter and leave the market at different times~\cite{Unver-2010,bloch-cantala,kurino-2014}.

The line of research that concerns questions akin to~\Qiv{} and~\Qv{} 
was initiated by Balinski and S\"onmez
in their paper on 
the property of respecting improvement in the context of college admission~\cite{balinski-sonmez}.
They proved that the student-optimal stable matching algorithm respects the improvement of students, so a better test score for a student always results in an outcome weakly preferred by the student (assuming other students' scores remain the same)---this means that the analog of question~\Qiv{} for the college admission problem (when viewed from the students' side) can always be answered affirmatively.
Hatfield et al.~\cite{hatfield-kojima-narita} contrasted the findings of Balinski and S\"onmez by showing that no stable mechanism respects the improvement of school quality.
S\"onmez and Switzer~\cite{sonmez-switzer} applied the model of \emph{matching with contracts} to the problem of cadet assignment 
in the United States Military Academy, and have proved that the cadet-optimal stable mechanism respects improvement of cadets.
Recently, Klaus and Klijn~\cite{klaus-klijn-RI} have obtained results of a similar flavor  in a school-choice model with minimal-access rights.

Roth et al.~\cite{roth-sonmez-unver-2005} deal with the property of respecting improvement in connection with kidney exchange:
they show that in a setting with dichotomous preferences and pairwise exchanges priority mechanisms are donor monotone, meaning that
a patient can only benefit from bringing an additional donor on board.

Closest to our work is the paper by
Bir\'o et al.~\cite{BKKV-mor} who focused on the classical Shapley--Scarf model and investigated 
how different solution concepts behave when the value of an agent's house increases. 
They proved that both the strict core and the set of competitive allocations satisfy the property of respecting improvements, 
although this is no longer true when the lengths of trading cycles are bounded by some constant. 
We remark that Bir\'o et al.~\cite{BKKV-mor} were not able to show that the property of respecting improvement holds for the core of housing markets. In fact, they posed questions~\Qiv{} and~\Qvi{} as open problems. We answer both of these questions affirmatively.

\subsection{Organization.}
Section~\ref{sec:prelim} contains all definitions necessary for our model.
Section~\ref{sec:arc-restrictions} deals with the decision problems associated with questions \Qi{}, \Qii{}, and \Qiii{}, and their computational complexity. 
In Section~\ref{sec:improvement-HM} we present our results on the property of respecting improvements in relation to the core of housing markets, that is, questions~\Qiv{} and~\Qv{}: Sections~\ref{sec:main-algo} and~\ref{sec:main-algo-correctness} contain our main technical result, Theorem~\ref{thm:core-main-RI}, 
while
Section~\ref{sec:strict-imp} deals with the computational complexity of the decision problem associated with question~\Qv{}.
In Section~\ref{sec:improvement-SR} we 
study the respecting improvement property in the context of \textsc{Stable Roommates}, i.e., question~\Qvi{}.
We close in Section~\ref{sec:future} with some questions for future research.

In an appendix we further provide some loosely related results: Appendix~\ref{sec:app-ttc} contains an adaptation of the TTC algorithm 
for partially ordered preferences.
Appendix~\ref{sec:app-strictcore} deals with the variants of questions~\Qi{}--\Qiii{} for the strict core in a setting where agents' preferences are weak orders.
Finally, 
Appendix~\ref{sec:app-max-core-approx}
contains an inapproximability result on the problem of maximizing the number of agents involved in trading in some core allocation.

\section{Preliminaries.}
\label{sec:prelim}

Here we describe our model, and provide all the necessary notation. 
Information about the organization of this paper can be found at the end of this section.

\subsection{Preferences as partial orders.}
In the majority of the existing literature, preferences of agents are usually considered to be either strict or, 
if the model allows for indifference, weak linear orders. Weak orders can be described as lists
containing \emph{ties}, a set of alternatives considered equally good for the agent. 
Partial orders are a generalization of weak orders that allow for two alternatives to be \emph{incomparable} for an agent. Incomparability may not be transitive, as opposed to indifference in weak orders.
Formally, an (irreflexive)\footnote{Throughout the paper we will use the term \emph{partial ordering} in the sense of an irreflexive (or strict) partial ordering.} \emph{partial ordering}~$\prec$ on a set of alternatives is an irreflexive, antisymmetric and transitive relation. 

Partially ordered preferences arise by many natural reasons; we give two examples motivated by kidney exchanges. For example, agents may be indifferent between goods that differ only slightly in quality.
Indeed, recipients might be indifferent between two organs if their expected graft survival times differ by less than one year.
However, small differences may add up to a significant contrast: an agent may be indifferent between $a$ and~$b$, and also  between $b$ and~$c$, but strictly prefer $a$ to~$c$. Such preferences result in so-called \emph{semiorders}, a special case of partial orders.

Partial preferences also emerge in multiple-criteria decision making. 
The two most important factors for estimating the quality of a kidney transplant are the HLA-matching between donor and recipient, and the age of the donor.\footnote{In fact, these are the two factors for which patients in the UK program can set acceptability thresholds \cite{biro-kidney-exchange-survey}.}
An organ is considered better than another if it is better with respect to both of these factors, leading to partial orders.

\subsection{Housing markets.}
Let $H=(N,\{\prec_a\}_{a \in N})$ be a \emph{housing market} with agent set~$N$ and with the preferences of each agent~$a \in N$ 
represented by a partial ordering~$\prec_a$ of the agents. For agents $a$, $b$, and $c$, 
we will write $a \preceq_c b$ as equivalent to $b \not\prec_c a$, and we write $a \sim_c b$ if $a \not\prec_c b$ and $b \not\prec_c a$.
We interpret $a \prec_c b$ (or~$a \preceq_c b$) as agent~$c$ \emph{preferring} (or \emph{weakly preferring}, respectively) the house owned by agent~$b$ to the house of agent~$a$.
We say that agent~$a$ finds the house of~$b$ \emph{acceptable}, if $a \preceq_a b$, 
and we denote by~$A(a)=\{b \in N: a \preceq_a b \}$ the set  of agents whose house is acceptable
for~$a$.
We define the \emph{acceptability graph} of the housing market~$H$ as the directed graph $G^H=(N,E)$ 
with $E=\{(a,b) \mid b \in A(a)\}$; we let  
$|G^H|=|N|+|E|$. Note that $(a,a) \in E$ for each~$a\in N$.
The \emph{submarket} of~$H$ on a set~$W \subseteq N$  of agents is the housing market $H_W=(W,\{ \prec_a^{|_{W}} \}_{a \in W})$ where $\prec_a^{|_{W}}$ is the partial order~$\prec_a$ restricted to~$W$;
the acceptability graph of~$H_W$ is the subgraph of~$G^H$ induced by~$W$, denoted by~$G^H[W]$.
For a set~$W$ of agents,  let $H-W$ be the submarket $H_{N \setminus W}$ obtained by \emph{deleting} $W$ from~$H$; for $W=\{a\}$ we may write simply $H-a$.

For a set $X \subseteq E$ of arcs in~$G^H$ and an agent~$a \in N$ 
we let $X(a)$ denote the set of agents~$b$ such that $(a,b) \in X$;
whenever $X(a)$ is a singleton~$\{b\}$ we will abuse notation by writing $X(a)=b$.
We also define $\delta^-_X(a)$ and $\delta^+_X(a)$ as the number of in-going and out-going arcs of~$a$ in~$X$, respectively.
For a set $W \subseteq N$ of agents, we let $X[W]$ denote the set of arcs in~$X$ that run between agents of~$W$.

We define an \emph{allocation}~$X$ in $H$ as a subset~$X \subseteq E$ of arcs in $G^H$ 
such that $\delta^-_X(a)=\delta^+_X(a)=1$ for each~$a \in N$, that is,
$X$ forms a collection of cycles in~$G^H$ containing each agent exactly once.
Then $X(a)$ denotes the agent whose house $a$ obtains according to allocation~$X$. 
If $X(a) \neq a$, then $a$ is \emph{trading} in~$X$.
For allocations~$X$ and~$X'$, we say that $a$ \emph{prefers} $X$ to~$X'$ if $X'(a) \prec_a X(a)$.

For an allocation $X$ in $H$, an arc~$(a,b) \in E$ is \emph{$X$-augmenting}, if $X(a) \prec_a b$. 
We define the \emph{envy graph}~$G^H_{X \prec}$ of~$X$ as the subgraph of~$G^H$ containing all $X$-augmenting arcs.
A \emph{blocking cycle} for~$X$ in~$H$ is a cycle in~$G^H_{X \prec}$, 
that is, a cycle~$C$ where each agent~$a$ on~$C$  prefers $C(a)$ to~$X(a)$.
An allocation~$X$ is contained in the \emph{core} of~$H$, if there does not exist a blocking cycle for it, i.e., if $G^H_{X \prec}$ is acyclic.
A \emph{weakly blocking cycle} for~$X$ is a cycle~$C$ in~$G^H$ where $X(a) \preceq_a C(a)$ for each agent~$a$ on~$C$
and $X(a) \prec_a C(a)$ for at least one agent~$a$ on~$C$. 
The \emph{strict core} of~$H$ contains allocations that do not admit weakly blocking cycles.

\section{Allocations in the core with arc restrictions.}
\label{sec:arc-restrictions}

We focus on the problem of finding an allocation in the core that fulfills certain arc constraints. 
The simplest such constraints arise when we require a given arc to be included in, 
or conversely, be avoided by
the desired allocation. 

The input of the \ec\ problem is a housing market $H=(N,\{\prec_a\}_{a \in N})$ 
and an arc $(a,b)$ in~$G^H$, 
and its task is to decide whether there exists an allocation in the core of~$H$ that contains~$(a,b)$, 
or in other words, where agent~$a$ obtains the house of agent~$b$.
Analogously, 
the \fec{} problem asks to decide if there exists an allocation 
in the core of~$H$ \emph{not} containing~$(a,b)$.

By giving a reduction from \textsc{Acyclic Partition}~\cite{BokalEtAl-2002}, we show in Theorem~\ref{thm:arc-in-core}
that both of these problems are computationally intractable, 
even if each agent has a strict ordering over the houses. 
In fact, we cannot even hope to decide for a given agent~$a$ in a housing market~$H$ whether there exists an allocation in the core of~$H$ 
where $a$ is trading; we call this problem \textsc{Agent Trading in Core}.
These results are in stark contrast to the polynomial-time solvability of the problem of finding a stable matching with forced and forbidden edges in an instance of \textsc{Stable Roommates}~\cite{Fleiner-Irving-Manlove}. 

\begin{theorem}
\label{thm:arc-in-core}
Each of the following problems is $\mathsf{NP}$-complete, even if agents' preferences are strict orders:
\begin{itemize}
\vspace{-6pt}
\item \ec{}, 
\item \fec{}, and 
\item \textsc{Agent Trading in Core}.
\end{itemize}
\end{theorem}

\begin{proof}
It is easy to see that all of these problems are in $\mathsf{NP}$, since given an allocation $X$ for~$H$, we can check in linear time whether it admits a blocking cycle: 
taking the envy graph $G_{X \prec}^H$ of $X$, we only have to check that it is \emph{acyclic}, i.e., contains no directed cycles 
(this can be decided using, e.g., some variant of the depth-first search algorithm).

To prove the $\mathsf{NP}$-hardness of \ec, we present a polynomial-time reduction from the \acycpart\ problem: given a directed graph $D$, decide whether it is possible to partition the vertices of $D$ into two acyclic sets $V_1$ and $V_2$. Here, a set $W$ of vertices is \emph{acyclic}, if $D[W]$ is acyclic. This problem was proved to be $\mathsf{NP}$-complete by Bokal et al.~\cite{BokalEtAl-2002}.

Given our input graph~$D$ with vertex set~$V$ and arc set~$A$, we construct a housing market $H$ as follows (see Fig.~\ref{fig:arc-in-core} for an illustration). 
We denote the vertices of $D$ by $v_1, \dots, v_n$, and we define the set of agents in $H$ 
as $$N=\{a_i,b_i,c_i,d_i \mid i \in \{1, \dots, n\} \cup \{a^\star,b^\star,a_0,b_0\}.$$
The preferences of the agents' are as shown below; for each agent $a \in N$ we only list those agents whose house $a$ finds acceptable. 
Here, for 
any set $W$ of agents we let $[W]$ denote an arbitrary fixed ordering of $W$.

  \begin{equation*}
    \begin{array}{lll}
    a^\star: & b^\star; 	& \\
    b^\star: & a_0, a_1, \dots, a_n, a^\star; & \\
    a_i: & b_i,b^\star 	& \textrm{ where $i \in \{0,1, \dots, n\}$};\\
    b_i: & c_{i+1}, d_{i+1} & \textrm{ where $i \in \{0,1, \dots, n-1\}$}; \\
    b_n: & a_0; 			& \\
    c_i: & d_i, [\{c_j \mid (v_i,v_j) \in A\}], a_i 			& \textrm{ where $i \in \{1, \dots, n\}$};\\
    d_i: & c_i, [\{d_j \mid (v_i,v_j) \in A\}], a_i \qquad \qquad 			& \textrm{ where $i \in \{1, \dots, n\}$}.\\
    \end{array}
  \end{equation*}

\begin{figure}[t]
\begin{center}
\includegraphics[scale=0.9]{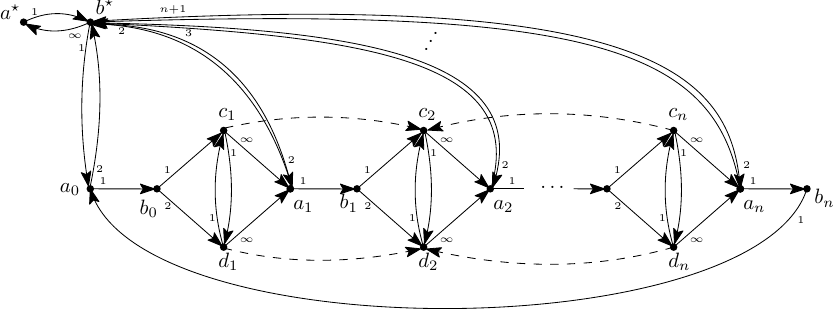}
\caption{Illustration of the housing market $H$ constructed in the $\mathsf{NP}$-hardness proof for \ec. 
Here and everywhere else we depict markets through their acceptability graphs with all loops omitted.
Preferences are indicated by numbers along the arcs;
the symbol~$\infty$ indicates the least-preferred choice of an agent. 
The example assumes that $(v_1,v_2)$ and $(v_n,v_2)$ are arcs of the directed input graph $D$, as indicated by the dashed arcs.}
\label{fig:arc-in-core}
\end{center}
\end{figure}

We finish the construction by defining our instance of \ec\ as the pair $(H,(a^\star,b^\star))$.
We claim that there exists an allocation in the core of $H$ containing $(a^\star,b^\star)$ if and only if the vertices of~$D$ can be partitioned into two acyclic sets.

``$\Rightarrow$'': Let us suppose that there exists an allocation $X$ that does not admit any blocking cycles and contains $(a^\star,b^\star)$. 

We first show that $X$ contains every arc $(a_i,b_i)$ for $i\in \{0,1,\dots, n\}$. To see this, observe that the only possible cycle in $X$ that contains $(a^\star,b^\star)$ 
is the cycle $(a^\star,b^\star)$ of length 2, because the arc $(b^\star,a^\star)$ is the only arc going into~$a^\star$. 
Hence, if for some $i \in \{0,1,\dots,n\}$ the arc $(a_i,b_i)$ is not in $X$, then the cycle $(a_i,b^\star)$ is a blocking cycle.
As a consequence, exactly one of the arcs~$(b_i,c_{i+1})$ and~$(b_i,d_{i+1})$ must be contained in $X$ for any $i \in \{0,1,\dots, n-1\}$, and similarly,
exactly one of the arcs~$(c_i,a_i)$ and~$(d_i,a_i)$ is contained in $X$ for any $i \in \{1,\dots, n\}$.

Next consider the agents $c_i$ and $d_i$ for some $i \in \{1,\dots, n\}$. 
As they are each other's top choice, it must be the case that either $(c_i,d_i)$ or $(d_i,c_i)$ is contained in $X$, 
as otherwise they both prefer to trade with each other as opposed to their allocation according to $X$, and the cycle $(c_i,d_i)$ would block $X$. 
Using the facts of the previous paragraph, we obtain that for each $v_i \in V$ exactly one of the following conditions holds:
\begin{itemize}
\item $X$ contains the arcs $(b_{i-1},c_i)$, $(c_i,d_i)$, and $(d_i,a_i)$, in which case we put $v_i$ into $V_1$;
\item $X$ contains the arcs $(b_{i-1},d_i)$, $(d_i,c_i)$, and $(c_i,a_i)$, in which case we put $v_i$ into $V_2$.
\end{itemize}

We claim that both $V_1$ and $V_2$ are acyclic in $D$. 
For a contradiction, let $C_1$ be a cycle within vertices of $V_1$ in $D$. 
Note that any arc $(v_i,v_j)$ of $C_1$ corresponds to an arc $(d_i,d_j)$ in the acceptability graph $G=G^H$ for $H$. 
Moreover, since $v_i \in V_1$, by definition we know that $d_i$ prefers $d_j$ to $X(d_i)=a_i$.
This yields that the agents $\{d_i \mid v_i \textrm{ appears on } C_1\}$ form a blocking cycle for $H$.
The same argument works to show that any cycle $C_2$ within $V_2$ corresponds to a blocking cycle formed by the agents in~$\{c_i \mid v_i \textrm{ appears on } C_2\}$,
proving the acyclicity of $V_2$.

``$\Leftarrow$'': Assume now that $V_1$ and $V_2$ are two acyclic subsets of $V$ forming a partition. 
We define an allocation $X$ to contain the cycle $(a^\star,b^\star)$, and a cycle consisting of the arcs in 
\begin{eqnarray*}
  X_{\circ} &=&\{(b_n,a_0)\} \cup \{(a_i,b_i) \mid v \in \{0,1,\dots,n\}\} \\
&&	\cup \, \{(b_{i-1},c_i),(c_i,d_i),(d_i,a_i) \mid v_i \in V_1\} \\
&&	\cup \, \{(b_{i-1},d_i),(d_i,c_i),(c_i,a_i) \mid v_i \in V_2\}.
\end{eqnarray*}
Observe that $X_{\circ}$ is indeed a cycle, and that $X$ is an allocation containing the arc $(a^\star,b^\star)$. 
We claim that the core of $H$ contains $X$. Assume for the sake of contradiction that $X$ admits a blocking cycle~$C$. 
Now, since $a^\star$, as well as each agent $a_i$, $i \in \{0,1,\dots, n\}$, is allocated its first choice by~$X$, none of these agents appears on $C$. 
This implies that neither $b^\star$, nor any of the agents $b_i$, $i \in \{0,1,\dots, n\}$, appears on $C$, 
since these agents have no in-neighbors that could possibly appear on $C$. Furthermore, every agent in $\{c_i \mid v_i \in V_1 \} \cup \{d_i \mid v_i \in V_2\}$
 is allocated its first choice by $X$. It follows that $C$ may contain only agents from $D_1=\{d_i \mid v_i \in V_1 \}$ and $C_2=\{c_i \mid v_i \in V_2 \}$.
 Observe that there is no arc in~$G$ from~$D_1$ to~$C_2$ or vice versa, hence $C$ is either contained in~$G[D_1]$ or~$G[C_2]$. 
Now, since any cycle within $G[D_1]$ or~$G[C_2]$ would correspond to a cycle in~$D$, 
the acyclicity of~$V_1$ and~$V_2$ ensures that $X$ admits no blocking cycle, proving the correctness of our reduction for the \ec\ problem.

Observe that the same reduction proves the $\mathsf{NP}$-hardness of \textsc{Agent Trading in Core}, 
since agent~$a^\star$ is trading in an allocation~$X$ for~$H$ if and only if the arc~$(a^\star,b^\star)$ is used in~$X$.

Finally, we modify the above construction to give a reduction from \acycpart\ to \fec. 
We simply add a new agent~$s^\star$ to $H$, with the house of~$s^\star$ being acceptable only for~$a^\star$ as its second choice (after~$b^\star$), and with
$s^\star$ preferring only $a^\star$ to its own house. We claim that the resulting market $H'$ together with the arc $(a^\star,s^\star)$ is a yes-instance of \fec\
if and only if $H$ with $(a^\star,b^\star)$ constitutes a yes-instance of \ec. 
To see this, it suffices to observe that any allocation for~$H'$ not containing $(a^\star,s^\star)$ is either blocked by the cycle $(a^\star,s^\star)$ of length 2, 
or contains the arc~$(a^\star,b^\star)$.
Hence, any allocation in the core of $H'$ contains  $(a^\star,b^\star)$ if and only if it does not contain $(a^\star,s^\star)$, proving the theorem.
\qed
\end{proof}

Theorem~\ref{thm:arc-in-core} shows that there is a computational gap between the strict core and the core: even though all three problems considered in Theorem~\ref{thm:arc-in-core} are $\mathsf{NP}$-complete even if agents' preferences are strict, the corresponding problems become computationally tractable for the strict core. 
This is trivial if the preferences are strict, since in that case the strict core contains a unique allocation~\cite{roth-postlewaite}.  
If preferences are weak orders (that is, if each agent orders the houses she finds acceptable in a linear order allowing ties), then the set of houses an agent can obtain in a strict core allocation can be computed in polynomial time based on the characterization of the strict core by Quint and Wako~\cite{WakoQuint}; see Appendix~\ref{sec:app-strictcore} for details. We remark that  since the characterization of  Quint and Wako crucially depends on weak orders, the same approach does not work when preferences are partial orders.

\section{The effect of improvements in housing markets.}
\label{sec:improvement-HM}

Let $H=(N,\{\prec_a\}_{a \in N})$ be a housing market containing agents $p$ and $q$. 
We consider a situation where the preferences of $q$ are modified by ``increasing the value'' of~$p$ for~$q$ without altering the preferences of~$q$ over the remaining agents.
If the preferences of~$q$ are given by a strict or weak order, then this translates to \emph{shifting} the position of~$p$ in the preference list of~$q$ towards the top. 
Formally, a housing market $H'=(N,\{\prec'_a\}_{a \in N})$ is called a $(p,q)$-\emph{improvement} of $H$, 
if $\prec_a=\prec'_a$ for any $a \in N \setminus \{q\}$, and~$\prec'_q$ is 
such that 
(i) $a \prec'_q b$ if and only if $a \prec_q b$ for each  $a,b \in N\setminus \{p\}$, 
and (ii) if $a \prec_q p$, then $a \prec'_q p$ for each $a \in N$.
We will also say that a housing market is a \emph{$p$-improvement} of~$H$, 
if it can be obtained by a sequence of $(p,q_i)$-improvements for a series $q_1, \dots, q_k$
of agents for some~$k \in \mathbb{N}$.

To examine how $p$-improvements affect the situation of~$p$ in the market, 
one may consider several solution concepts such as the core, the strict core, and so on.
We regard a solution concept as a function $\Phi$ that assigns a set of allocations to each housing market. 
Based on the preferences of~$p$, we can compare allocations in~$\Phi$.
Let $\Phi^+_p(H)$ denote the set containing the best houses $p$ can obtain in~$\Phi(H)$:
$$\Phi^+_p(H)=\{X(p) \mid X \in \Phi(H), \forall X' \in \Phi(H): X'(p) \preceq_p X(p)\}.$$
Similarly, let $\Phi^-_p(H)$ be the set containing the worst houses $p$ can obtain in~$\Phi(H)$.

Following the notation used by Bir\'o et al.~\cite{BKKV-mor}, 
we say that $\Phi$ \emph{respects improvement for the best available house} or simply \emph{satisfies the RI-best property}, 
if for any housing markets~$H$ and~$H'$ such that $H'$ is a $p$-improvement of~$H$ for some agent~$p$,
$a \preceq_p a'$ for every $a \in \Phi^+_p(H)$ and $a' \in \Phi^+_p(H')$.
Similarly, $\Phi$ \emph{respects improvement for the worst available house} or simply \emph{satisfies the RI-worst property}, 
if for any housing markets~$H$ and~$H'$ such that $H'$ is a $p$-improvement of~$H$ for some agent~$p$,
$a \preceq_p a'$ for every $a \in \Phi^-_p(H)$ and $a' \in \Phi^-_p(H')$.

Notice that the above definition does not take into account the possibility that 
a solution concept~$\Phi$ may become empty as a result of a $p$-improvement. 
To exclude such a possibility, we may require the condition that an improvement does not destroy all solutions.
We say that $\Phi$ \emph{strongly satisfies the RI-best (or RI-worst) property}, if 
besides satisfying the RI-best (or, respectively, RI-worst) property, it also guarantees that 
whenever $\Phi(H) \neq \emptyset$, then $\Phi(H') \neq \emptyset$ also holds where $H'$ is a $p$-improvement of~$H$ for some agent~$p$.

We prove that the core of housing markets strongly satisfies the RI-best property. 
In fact, Theorem~\ref{thm:core-main-RI} (proved in Section~\ref{sec:main-algo-correctness}) 
states a slightly stronger statement. 

\begin{theorem}
\label{thm:core-main-RI}
For any allocation $X$ in the core of housing market $H$ and a $p$-improvement~$H'$ of~$H$, 
there exists an allocation $X'$ in the core of $H'$ such that either $X(p)=X'(p)$ or $p$ prefers $X'$ to~$X$. 
Moreover, given $H$, $H'$ and $X$, it is possible to find such an allocation $X'$ in $O(|H|)$ time.
\end{theorem}

\begin{corollary}
\label{cor:possible-houses}
The core of housing markets strongly satisfies the RI-best property.
\end{corollary}

By contrast, we show that the RI-worst property does not hold for the core. 

\begin{proposition}
\label{prop:core-worst-RI-fails}
The core of housing markets violates the RI-worst property, even if agents' preferences are strict orders.
\end{proposition}

\begin{figure}[t]
\begin{center}
\includegraphics[scale=1]{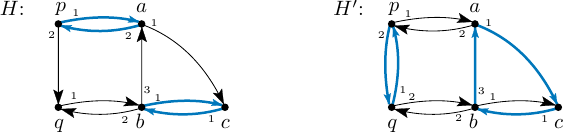}
\caption{The housing markets $H$ and $H'$ in the proof of Proposition~\ref{prop:core-worst-RI-fails}.
For both $H$ and $H'$, the allocation represented by bold (and blue) arcs yields the worst possible outcome for~$p$ 
in any core allocation of the given market. 
}
\label{fig:example-core-RI-worst}
\end{center}
\end{figure}

\begin{proof}
Let $N=\{a,b,c,p,q\}$ be the set of agents. The preferences indicated in Figure~\ref{fig:example-core-RI-worst} 
define a housing market~$H$ and a~$(p,q)$-improvement $H'$ of~$H$. 

We claim that in every allocation in the core of~$H$, agent~$p$ obtains the house of~$a$.
To see this, let~$X$ be an allocation where $(p,a) \notin X$. If agent~$a$ is not trading in~$X$, then $a$ and~$p$ form a blocking cycle;
therefore, $(b,a) \in X$. Now, if $(c,b) \notin X$, then $c$ and~$b$ form a blocking cycle for~$X$;
otherwise, $q$ and~$b$ form a blocking cycle for~$X$.  Hence, $p$ obtains her top choice in all core allocations  of~$H$.

However, it is easy to verify that the core of $H'$ contains an allocation where $p$ obtains only her second choice ($q$'s house), 
as shown in Figure~\ref{fig:example-core-RI-worst}.
\qed
\end{proof}

We remark that Corollary~\ref{cor:possible-houses} and Proposition~\ref{prop:core-worst-RI-fails} illuminate both the similarities of and the contrast between the properties of the core and the strict core.
Recall that for strict preferences, there is a unique allocation in the strict core, and in this case the results of Bir\'o et al.~\cite{BKKV-mor} show that the strict core strongly satisfies both the RI-best and the RI-worst properties. 
On the one hand, Proposition~\ref{prop:core-worst-RI-fails} hence shows a sharp difference between the core and the strict core with respect to the RI-worst property. 
On the other hand, Corollary~\ref{cor:possible-houses} is very close to the analogous findings by Bir\'o et al~\cite{BKKV-mor}: both results establish that the given solution concept (the core or the strict core) satisfies the RI-best property.
See also Table~\ref{tab:summary-RI} in Section~\ref{sec:future} for a comparison of the core and the strict core in relation to the RI-best and RI-worst properties. We remark that despite the proximity of Corollary~\ref{cor:possible-houses} with the results by Bir\'o et al.~\cite{BKKV-mor} for the strict core, they are independent in the sense that neither of them implies the other.

\medskip
We describe our algorithm for proving Theorem~\ref{thm:core-main-RI} in Section~\ref{sec:main-algo}, 
and prove its correctness in Section~\ref{sec:main-algo-correctness}.
In Section~\ref{sec:strict-imp} we look at the problem of deciding whether 
a $p$-improvement leads to a situation strictly better for~$p$.

\subsection{Description of algorithm \algoHM{}.}
\label{sec:main-algo}

Before describing our algorithm for Theorem~\ref{thm:core-main-RI}, we need some notation.

\paragraph{Sub-allocations and their envy graphs.}
Given a housing market $H=(N,\{\prec_a\}_{a \in N})$ and two subsets~$U$ and~$V$ of agents in $N$ with $|U|=|V|$, 
we say that a set $Y$ of arcs in $G^H=(N,E)$ is a \emph{sub-allocation} from~$U$ to~$V$ in~$H$, if 
\begin{itemize}
\item[$\bullet$] $\delta^+_Y(v)=0$  for each $v \in V$, and $\delta^+_Y(a)=1$ for each $a \in N \setminus V$; 
\item[$\bullet$] $\delta^-_Y(u)=0$ for each $u \in U$, and $\delta^-_Y(a)=1$ for each $a \in N \setminus U$.
\end{itemize}
Note that $Y$ forms a collection of mutually vertex-disjoint cycles and paths $P_1, \dots, P_k$ in~$G^H$, 
with each path~$P_i$ leading from a vertex of~$U$ to a vertex of~$V$. 
Moreover, the number of paths in this collection is $k=|U \triangle V|$, 
where $\triangle$ stands for the symmetric difference operation. 
We call $U$ the \emph{source set} of~$Y$, and $V$ its \emph{sink set}. 
See Figure~\ref{fig:sub-allocation} for an illustration.

\begin{figure}
    \centering
    \includegraphics{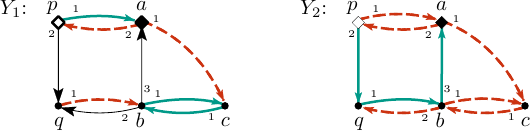}
    \caption{Illustration for the concept of sub-allocation. The arcs sets~$Y_1$ and $Y_2$, shown with bold, teal lines, are both sub-allocations from~$\{p\}$ to~$\{a\}$ in the depicted housing market  (as usual, loops are omitted).
    Source and sink vertices of~$Y$ are depicted with a white black diamond, respectively.
    For each of~$Y_1$ and~$Y_2$, we show the corresponding envy arcs (i.e., the arcs in the corresponding envy graphs) with dashed, red lines; as can be seen, $Y_1$ is stable while $Y_2$ is not.}
    \label{fig:sub-allocation}
\end{figure}

Given a sub-allocation $Y$ from $U$ to $V$ in $H$, an arc $(a,b) \in E$ is \emph{$Y$-augmenting}, 
if either $a \in V$ or~$Y(a) \prec_a b$. 
We define the \emph{envy graph} of $Y$ as $G^H_{Y \prec}=(N,E_Y)$ where $E_Y$ is the set of $Y$-augmenting arcs in $E$.
A blocking cycle for $Y$ is a cycle in~$G^H_{Y \prec}$.
We say that the sub-allocation~$Y$ is \emph{stable}, if no blocking cycle exists for~$Y$, that is, if its envy graph is acyclic.

\smallskip

We are now ready to propose an algorithm called \algoHM{} 
that given an allocation~$X$ in the core of~$H$ outputs an allocation~$X'$ as required by Theorem~\ref{thm:core-main-RI}.
Let $q_1, \dots, q_k$ denote the agents for which $H'$ can be obtained from~$H$ by a series of $(p,q_i)$-improvements, $i = 1, \dots, k$. 
Observe that we can assume w.l.o.g. that the agents $q_1, \dots, q_k$ are all distinct.

\paragraph{Algorithm \algoHM{}.}
For a pseudocode description, see Algorithm~\ref{alg:hm-improve}, and for an example demonstrating the algorithm see Example~\ref{ex:HMalgo}.

First, \algoHM{} checks whether $X$ belongs to the core of $H'$, and if so, outputs $X'=X$. 
Hence, we may assume that $X$ admits a blocking cycle in $H'$. 
Let $Q$ denote that set of only those agents among $q_1, \dots, q_k$ that in~$H'$ prefer $p$'s house to the one they obtain in allocation~$X$, 
that is, \[Q= \left\{q_i: X(q_i) \prec'_{q_i} p, 1 \leq i \leq k \right\}.\]
Observe that if an arc is an $X$-augmenting arc in~$H'$ but not in~$H$, then it must be an arc of the form $(q,p)$ where $q \in Q$.
Therefore
any cycle that blocks~$X$ in~$H'$ must contain an arc from $\{(q,p): q \in Q\}$, 
as otherwise it would block~$X$ in~$H$ as well. 

\algoHM{} proceeds by modifying the housing market: for each $q \in Q$, it adds a new agent~$\qq$ to~$H'$, 
with $\qq$ taking the place of~$p$ in the preferences of~$q$; 
the only house that agent~$\qq$ prefers to her own will be the house of~$p$ (the preferences of $p$ remain unchanged).
Let $\HH$ be the housing market obtained. 
Then the acceptability graph~$\GG$ of $\HH$ can be obtained from the acceptability graph  of~$H'$ 
by subdividing the arc $(q,p)$ for each $q \in Q$ with a new vertex corresponding to agent~$\qq$, i.e., replacing the arc~$(q,p)$ with arcs~$(q,\qq)$ and~$(\qq,p)$. For an illustration of the construction, see Figure~\ref{fig:EX-construction}.
Let~$\QQ=\{\qq:q \in Q\}$, $\NN=N \cup \QQ$, and let us denote by~$\EE$ be the set of arcs in~$\GG$.

\begin{figure}[t]
\begin{center}
\includegraphics[scale=1]{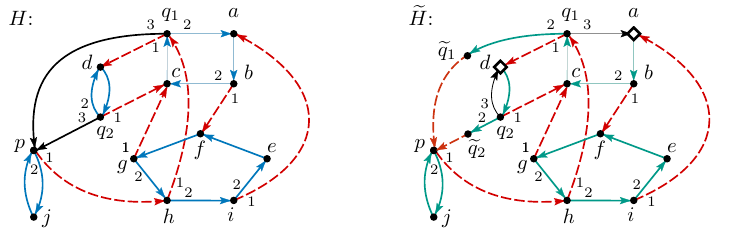}
\caption{The housing market~$H$ of Example~\ref{ex:HMalgo} and the modified housing market~$\HH$ constructed by Algorithm \algoHM{} based on the $p$-improvement of~$H$ where $q_1$ and $q_2$ change their preferences so that $q_1$ comes to prefer~$p$ to~$a$, 
 and $q_2$ comes to prefer $p$ to~$d$. We depicted the core allocation~$X$ for~$H$ using blue lines, and we depicted the corresponding sub-allocation~$Y$, as constructed by Algorithm \algoHM{} in its initialization step, using teal lines.
Sub-allocation~$Y$ has two sources, $a$ and~$d$, highlighted by diamonds, and two sinks,~$\qq_1$ and~$\qq_2$. Envy arcs for both the original allocation~$X$ in~$H$ and the sub-allocation~$Y$ in~$\HH$ are shown using red, dashed lines.
}
\label{fig:EX-construction}
\end{center}
\end{figure}

\medskip
{\bf Initialization.}
Let $Y=X \setminus \{ (q,X(q)): q \in Q \} \cup \{(q,\qq): q \in Q\}$ in $\GG$.
Observe that $Y$ is a sub-allocation in~$\HH$ with source set~$\{X(q):q \in Q\}$ and sink set~$\QQ$.
Additionally, we define a set $R$ of \emph{irrelevant} agents, initially empty.
We may think of irrelevant agents as temporarily deleted from the market. 

\medskip
{\bf Iteration.}
Next, algorithm~\algoHM{} iteratively modifies the sub-allocation~$Y$ and the set~$R$ of irrelevant agents.
It will maintain the property that $Y$ is a sub-allocation in~$\HH - R$;
we denote its envy graph by $\GG_{Y \prec}$, having vertex set $\NN \setminus R$. 
While the source set of~$Y$ changes quite freely during the iteration, the sink set always remains a subset of~$\QQ$.

At each iteration, \algoHM{} performs the following steps:
\smallskip
\begin{enumerate}[itemindent=40pt,itemsep=2pt]
\item[{\bf Step~1.}]  Let $U$ be the source set of $Y$, and $V$ its sink set. If $U=V$, then the iteration stops.
\item[{\bf Step~2.}]  Otherwise, if there exists a $Y$-augmenting arc~$(s,u)$ in~$\GG_{Y \prec}$ entering some source vertex~${u \in U}$ 
(note that $s \in \NN \setminus R$), then proceed as follows.
\begin{itemize}[itemindent=0pt,labelsep=2pt,labelwidth=6pt]
\item[{\bf (a)}] If $s \notin V$, then let $u'=Y(s)$.
The algorithm modifies $Y$ by deleting the arc~$(s,u')$ and adding the arc $(s,u)$ to~$Y$. 
Note that $Y$ thus becomes a sub-allocation from~$U \setminus \{u\} \cup \{u'\}$ to~$V$ in~$\HH - R$.
\item[{\bf (b)}] If $s \in V$, then simply add the arc $(s,u)$ to~$Y$. 
In this case $Y$ becomes a sub-allocation from~$U \setminus \{u\} $ to~$V \setminus \{s\}$ in~$\HH - R$.
\end{itemize}
\item[{\bf Step~3.}]  Otherwise, let $u$ be any vertex in~$U \setminus V$ (not entered by any arc in~$\GG_{Y \prec}$), and  
let $u'=Y(u)$. 
The algorithm adds $u$ to the set~$R$ of irrelevant agents,
and modifies $Y$ by deleting the arc~$(u,u')$.
Again, $Y$ becomes a sub-allocation from~$U \setminus \{u\} \cup \{u'\}$ to~$V$ in~$\HH - R$.
\end{enumerate}

\begin{figure}[t]
\begin{center}
\includegraphics[scale=1]{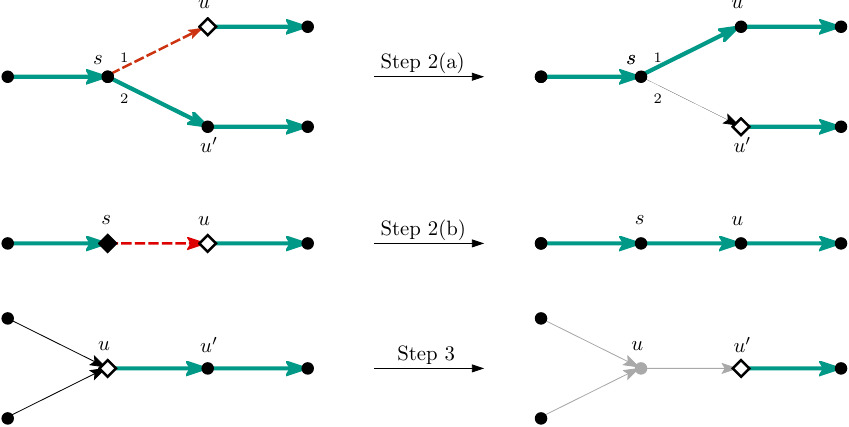}
\caption{Illustration of the possible steps performed during the iteration by \algoHM{}.
The edges of the current sub-allocation~$Y$ are depicted using bold, teal lines, while edges of the
envy graph~$\GG_{Y \prec}$ are shown by dashed, red lines. 
As in Figure~\ref{fig:sub-allocation}, source and sink vertices of~$Y$ are depicted with a white black diamond, respectively.
Vertices of~$R$ as well as all edges incident to them are shown in grey.
}
\label{fig:iteration}
\end{center}

\end{figure}

\medskip
{\bf Output.}
Let $Y$ be the sub-allocation at the end of the above iteration, $U=V$ its source \emph{and} sink set, and $R$ the set of irrelevant agents.
Note that $\QQ \setminus R \setminus U$ may contain at most one agent. 
Indeed, if~$\qq \in \QQ \setminus R \setminus U$, then $Y$ must contain the unique arc leaving~$\qq$,  namely $(\qq,p)$;
therefore, by $\delta^-_Y(p) \leq 1$, at most one such agent $\qq$ can exist. 

To construct the desired allocation~$X'$, the algorithm first 
applies the variant of the TTC algorithm that can deal with partial order preferences, described in Appendix~\ref{sec:app-ttc},  
to the submarket~$H'_{R \cap N}$ of~$H'$ when restricted to the set of irrelevant agents. 
This algorithm computes an allocation $X_R$ in the core of~$H'_{R \cap N}$.

\algoHM{} next 
deletes all agents in~$\QQ$. 
Since any agent in~$\QQ \cap U=\QQ \cap V=V$ has zero in- and outdegree in~$Y$, there is no need to modify our sub-allocation
when deleting such agents; the same applies to agents in~$\QQ \cap R$. 
By contrast, if there exists an agent~$\qq \in \QQ \setminus R \setminus U$, then 
$Y$ must contain the unique incoming and outgoing arcs of~$\qq$,
and therefore the algorithm replaces the arcs $(q,\qq)$ and $(\qq,p)$ with the arc~$(q,p)$.
This way we obtain an allocation on the submarket of~$H'$ on agents set $N \setminus R$. 

Finally, \algoHM{} outputs an allocation~$X'$ defined as
\[
X'= \left\{
\begin{array}{ll} 
X_R \cup Y  & \textrm{ if $\QQ \setminus R \setminus U=\emptyset$,} \\
X_R \cup Y \setminus \{(q,\qq), (\qq,p)\} \cup \{(q,p)\} \qquad\qquad  &  \textrm{ if $\QQ \setminus R \setminus U=\{\qq\}$.} 
\end{array}
\right.
\]

\begin{algorithm}
\caption{Algorithm \algoHM{}}\label{alg:hm-improve}
\hspace*{\algorithmicindent} \textbf{Input:} housing market~$H=(N,\prec)$, its $p$-improvement~$H'=(N,\prec')$ for some agent~$p$, and \\
\hspace*{\algorithmicindent} \phantom{\textbf{Input:}}
an allocation~$X$ in the core of~$H$. \\
\hspace*{\algorithmicindent} \textbf{Output:} an allocation~$X'$ in the core of~$H'$ such that $X(p) \prec_p X'(p)$ or $X(p)=X'(p)$.
\begin{algorithmic}[1]
\If{$X$ is in the core of~$H'$} \textbf{return}~$X$ \EndIf
\State Set $Q= \{a \in N: \prec_a \not = \prec'_a$ and $X(a) \prec'_a p\}$. 
\State Initialize housing market $\HH \leteq H$.
	\ForAll{$q \in Q$}
	\State Add new agent~$\qq$ to~$\HH$, preferring only~$p$ to her own house.
	\State Replace $p$ with $\qq$ in the preferences of~$q$ in~$\HH$.
	\EndFor
\State Set $\QQ=\{\qq : q \in Q\}$.
		\Comment{$\HH$ is now defined.}
\State Create sub-allocation $Y \leteq X \setminus \{ (q,X(q)): q \in Q \} \cup \{(q,\qq): q \in Q\}$.
\State Set $U$ and $V$ as the source and sink set of~$Y$, resp., and set $R \leteq \emptyset$.
\While{$U \neq V$}
\If{there exists an arc~$(s,u)$ in the envy graph~$\GG_{Y \prec}$ with $u \in U$}
		\If{$s \notin V$} 
		\State Set $u' \leteq Y(s)$, and update $Y \leftarrow Y \setminus \{(s,u')\} \cup \{(s,u)\}$ and $U \leftarrow U \setminus \{u\} \cup \{u'\}$.
		\Else \Comment{Case $s \in V$.}
		\State Update $Y \leftarrow Y \cup \{(s,u)\}$, $U \leftarrow U \setminus \{u\}$ and $V \leftarrow V \setminus \{s\}$.		
		\EndIf
\Else \Comment{No arc enters~$U$ in the envy graph~$\GG_{Y \prec}$.}
	\State Pick any agent~$u \in U \setminus V$, and set $u' \leteq Y(u)$. 
	\State Update $Y \leftarrow Y \setminus (u,u')$, $U \leftarrow U \setminus \{u\} \cup \{u'\}$ and $R \leftarrow R \cup \{u\}$.
\EndIf
\EndWhile
\State Compute a core allocation~$X_R$ in the submarket~$H'_{R \cap N}$.
\If{$\QQ \setminus R \setminus U=\emptyset$} set $X' \leteq X_R \cup Y$.
\Else{} set $X' \leteq X_R \cup Y \setminus \{(q,\qq), (\qq,p)\} \cup \{(q,p)\}$ where  $\QQ \setminus R \setminus U=\{\qq\}$.
\EndIf 
\State \textbf{return} the allocation $X'$.
\end{algorithmic}
\end{algorithm}

Let us now illustrate how Algorithm~\algoHM{} works on an example. 

\begin{example}
\label{ex:HMalgo}
Let us consider the housing market~$H$ shown in Figure~\ref{fig:EX-construction}, and let $X$ denote the allocation in the core of~$H$ depicted, i.e., $X$ consists of the cycles $(p,j)$, $(a,b,c,q_1)$, $(d,q_2)$, and~$(e,f,g,h,i)$. 
Consider now the $p$-improvement~$H'$ of~$H$ where both $q_1$ and~$q_2$ place~$p$ as their second favorite choice (instead of the third one). 

The algorithm starts by checking whether $X$ is in the core of~$H'$, and finds that---since arcs~$(q_1,p)$ and~$(q_2,p)$ have become $X$-augmenting arcs---allocation~$X$ admits the blocking cycle $(q_1,p,h)$. 
Thus, the algorithm proceeds with modifying the housing market by subdividing the arcs~$(q_1,p)$ and~$(q_2,p)$ with newly added agents~$\qq_1$ and~$\qq_2$. 

\edef\myindent{\the\parindent}
\edef\myparskip{\the\parskip}

\medskip
\noindent
\begin{minipage}{0.64\textwidth}
\setlength{\parskip}{\myparskip}
\setlength{\parindent}{\myindent}

In the initialization phase, algorithm \algoHM{} constructs the sub-allocation~$Y$ based on~$X$ in the modified housing market~$\HH$, as seen on Figure~\ref{fig:EX-construction}; we repeat the figure to the right here. Its source set is $U=\{a,d\}$ and its sink set is $V=\QQ=\{\qq_1,\qq_2\}$. The set of irrelevant agents is set to~$R=\emptyset$. 
Then the algorithm starts iterating Steps~1--3 with the following results: \par
\end{minipage}\qquad
\begin{minipage}{0.3\textwidth}
\includegraphics[width=\textwidth]{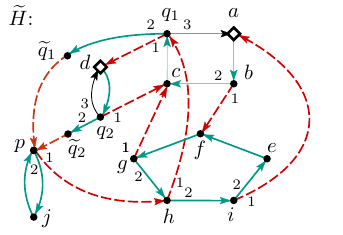}
\end{minipage}
\begin{enumerate}

\smallskip

\begin{minipage}{0.6\textwidth}
\item Considering the initial market~$\HH$ as shown in Figure~\ref{fig:EX-construction}, the algorithm finds that both sources, $a$ and~$d$, are entered by some $Y$-augmenting arc, namely by~$(q_1,d)$ and~$(i,a)$. It may choose either one of these arcs to proceed with; we consider the course of the algorithm when it starts with the arc~$(q_1,d)$: it replaces~$(q_1,\qq_1)$ with~$(q_1,d)$ in~$Y$, so the source set becomes $U=\{a,\qq_1\}$, while the sink set remains $V=\{\qq_1,\qq_2\}$. The resulting sub-allocation is depicted to the right. 
\end{minipage}\qquad
\begin{minipage}{0.3\textwidth}
\includegraphics[width=\textwidth]{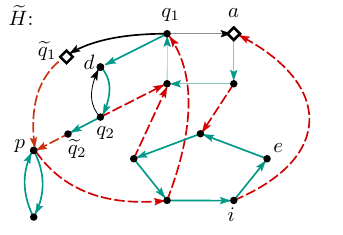}
\end{minipage}

\begin{minipage}{0.6\textwidth}
\item Next, the algorithm finds that only the source agent~$a$ (from among the source set~$U=\{a,\qq_1\}$) is entered by some $Y$-augmenting arc, namely, by~$(i,a)$. It replaces $(i,e)$ with~$(i,a)$ in~$Y$ yielding the sub-allocation shown to the right; the source set becomes $U=\{e,\qq_1\}$. 
\end{minipage}\qquad
\begin{minipage}{0.3\textwidth}
\includegraphics[width=\textwidth]{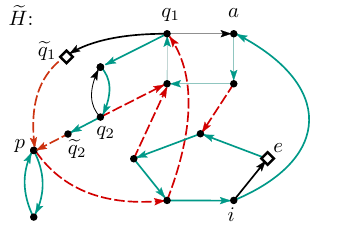}
\end{minipage} 

\begin{minipage}{0.6\textwidth}
\item Next, the algorithm finds that neither of the sources~$e$ and~$\qq_1$ is entered by an $Y$-augmenting arc; hence, it takes the unique source agent that is not a sink, namely $e$, and declares it irrelevant by setting $R=\{e\}$, and removing~$e$ from the market~$\HH$. The arcs of~$Y$ incident to~$e$ are removed from~$Y$, thus the source set becomes $U=\{f,\qq_1\}$.
\end{minipage}\qquad
\begin{minipage}{0.3\textwidth}
\includegraphics[width=\textwidth]{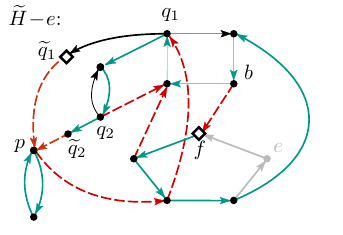}
\end{minipage} 

\begin{minipage}{0.6\textwidth}
\item Next, the algorithm finds that only the source agent~$f$ (from among the source set~$U=\{f,\qq_1\}$) is entered by some $Y$-augmenting arc, namely, by~$(b,f)$. It replaces $(b,c)$ with~$(b,f)$ in~$Y$, and the source set becomes $U=\{c,\qq_1\}$. 
\end{minipage}\qquad
\begin{minipage}{0.3\textwidth}
\includegraphics[width=\textwidth]{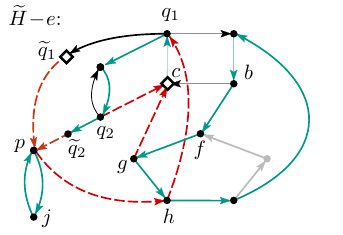}
\end{minipage} 

\begin{minipage}{0.6\textwidth}
\item Next, the algorithm finds that only the source agent~$c$ (from among the source set~$U=\{f,\qq_1\}$) is entered by some $Y$-augmenting arc, namely, the arcs~$(q_2,c)$ and~$(g,c)$. Here we consider the course of the algorithm when it chooses the arc~$(q_2,c)$: it replaces~$(q_2,\qq_2)$ with~$(q_2,c)$ in~$Y$, and the source set becomes $U=\{\qq_2,\qq_1\}$. 
\end{minipage}\qquad
\begin{minipage}{0.3\textwidth}
\includegraphics[width=\textwidth]{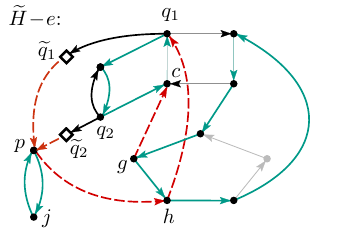}
\end{minipage} 
\end{enumerate}

At this point, the algorithm detects that the source set $U$ equals the sink set~$V=\QQ$, and stops the iteration. It computes a core allocation for the  submarket  $H'_{R \cap N}$ of irrelevant agents; since $R=\{e\}$, this allocation consists of the single arc~$(e,e)$.  Since $\QQ \setminus R \setminus U=\emptyset$, it outputs the allocation $Y \cup \{(e,e)\}$ in which agents trade along the cycles~$(p,j)$, $(q_1,c,q_2,d)$, and~$(a,b,f,g,h,i)$; see Figure~\ref{fig:EX-output}. 

\medskip

Consider now an alternative course for the algorithm when, after the fourth iteration (see the figure next to Step~4 above), in the fifth iteration step the arc~$(g,c)$ gets chosen instead of the arc~$(q_2,c)$; we omit the corresponding figures for Steps~$5^*$ and~$6^*$ in this alternative course:

\begin{enumerate}
\begin{minipage}{0.9\textwidth}
\item[$\quad$5$^*$.] The algorithm replaces~$(g,h)$ with~$(g,c)$ in~$Y$, and the source set becomes $U=\{h,\qq_1\}$. 
\end{minipage}

\begin{minipage}{0.9\textwidth}
\item[$\quad$6$^*$.] The algorithm replaces~$(p,j)$ with~$(p,h)$ in~$Y$, and the source set becomes $U=\{j,\qq_1\}$. 
\end{minipage}

\begin{minipage}{0.6\textwidth}
\item[7$^*$.] The algorithm finds that no $Y$-augmenting arc enters either of the sources, and thus removes agent~$j$, the only agent in~$U \setminus V$, together with the arc~$(j,p)$. Hence, the set of irrelevant agents is set to $R=\{e,j\}$, and 
the source set becomes $U=\{p,\qq_1\}$. 
\end{minipage}\qquad
\begin{minipage}{0.3\textwidth}
\includegraphics[width=\textwidth]{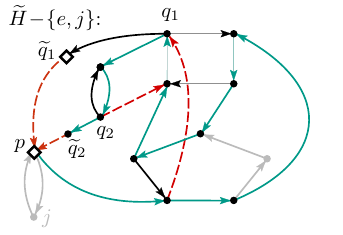}
\end{minipage} 

\begin{minipage}{0.6\textwidth}
\item[8$^*$.] The algorithm finds that only the source~$p$ is entered by some $Y$-augmenting arc, namely by the arcs~$(\qq_1,p)$ and~$(\qq_2,p)$. It chooses one of them, say~$(\qq_1,p)$. Since~$\qq_1$ is a sink, it adds~$(\qq_1,p)$ to~$Y$, removes~$p$ from the source set, and removes~$\qq_1$ from the sink set. This yields $U=\{\qq_1\}$ and $V=\{\qq_2\}$. 
\end{minipage}\qquad
\begin{minipage}{0.3\textwidth}
\includegraphics[width=\textwidth]{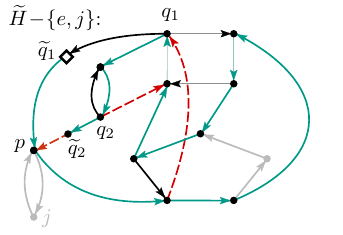}
\end{minipage} 

\begin{minipage}{0.6\textwidth}
\item[9$^*$.] The unique source~$\qq_1$ is not entered by any $Y$-augmenting arc, hence the algorithm declares it irrelevant by setting~$R=\{e,j,\qq_1\}$ and removes it form the market. The arc~$(\qq_1,p)$ is removed from~$Y$, and the source set becomes $U=\{p\}$. 
\end{minipage}\qquad
\begin{minipage}{0.3\textwidth}
\includegraphics[width=\textwidth]{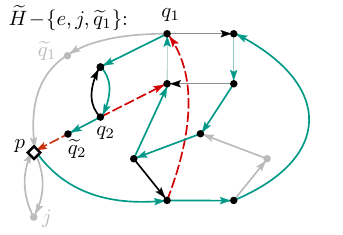}
\end{minipage} 

\begin{minipage}{0.6\textwidth}
\item[10$^*$.] The algorithm finds that the (unique) source~$p$ is entered by a unique $Y$-augmenting arc, namely~$(\qq_2,p)$. Since~$\qq_2$ is a sink (recall that $V=\{\qq_2\}$ at this point), it adds $(\qq_2,p)$ to~$Y$, removes~$p$ from the source set, and removes~$\qq_2$ from the sink set. This yields $U=V=\emptyset$. 
\end{minipage}\qquad
\begin{minipage}{0.3\textwidth}
\includegraphics[width=\textwidth]{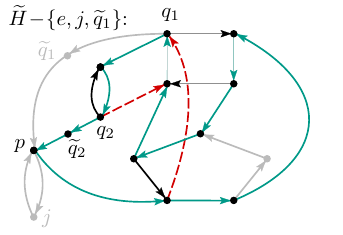}
\end{minipage} 
\end{enumerate}

At this point, the algorithm detects that the source set $U$ equals the sink set~$V$ (both empty), and stops the iteration. It computes a core allocation for the submarket $H'_{R \cap N}$ of irrelevant agents; since $R=\{e,j,\qq_1\}$, this allocation consists of the arcs~$(e,e)$ and~$(j,j)$. Since $\QQ \setminus R \setminus U=\{\qq_2\}$, it outputs the allocation $Y \setminus \{(q_2,\qq_2),(\qq_2,p) \} \cup \{(q_2,p\} \cup \{(e,e),(j,j)\}$ in which agents trade along the single cycle~$(p,h,i,a,b,f,g,c,q_1,d,q_2)$; see Figure~\ref{fig:EX-output}. 
\end{example}

\begin{figure}[b]
\begin{center}
\includegraphics[scale=1]{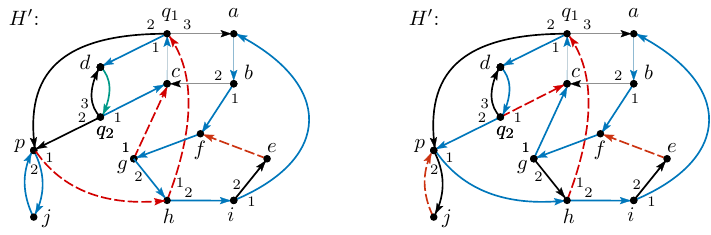}
\caption{Two allocation in the core of housing market~$H'$ of Example~\ref{ex:HMalgo}, computed by Algorithm \algoHM{}. The  figure to the left depicts the allocation obtained by Steps~1--5, while the figure to the right depicts the allocation obtained when Steps~1--4 are followed by Steps~$5^*$--$10^*$.
}
\label{fig:EX-output}
\end{center}
\end{figure}

\subsection{Correctness of algorithm \algoHM{}.}
\label{sec:main-algo-correctness}

We begin proving the correctness of algorithm \algoHM{} with the following.

\begin{lemma} 
\label{lem:stable-pre-allocation}
At each iteration, sub-allocation~$Y$ is stable in $\HH -R$.
\end{lemma}

\begin{proof}
The proof is by induction on the number $n$ of iterations performed.
For $n=0$, suppose for the sake of contradiction that $C$ is a cycle in~$\GG_{Y \prec}$. 
First note that $C$ cannot contain any agent in~$\QQ$, since the unique arc entering~$\qq$, that is, the arc~$(q,\qq)$, 
is contained in~$Y$ by definition. Hence, $C$ is also a cycle in~$H$.
Moreover, recall that initially $Y(a)=X(a)$ for each agent $a \in N \setminus Q$, and 
by the definition of~$Q$, we also know $X(q) \prec_q \qq=Y(q)$ for each $q \in Q$. 
Therefore, any arc of~$C$ is an $X$-augmenting arc as well, and thus $C$ is a blocking cycle for~$X$ in~$H$.
This contradicts our assumption that $X$ is in the core of~$H$. 
Hence, $Y$ is stable in~$\HH$ at the beginning; note that $R=\emptyset$ initially. 

For $n \geq 1$, assume that the algorithm has performed $n-1$ iterations so far. 
Let $Y$ and $R$ be as defined at the beginning of the $n$-th iteration, and 
let $Y'$ and $R'$ be the sub-allocation and the set of irrelevant agents obtained after the modifications in this iteration. 
Let also $U$ and~$V$ ($U'$ and~$V'$) denote the source and sink set of~$Y$ (of~$Y'$, respectively).
By induction, we may assume that $Y$ is stable in $\HH -R$, so $\GG_{Y \prec}$ is acyclic.
In case \algoHM{} does not stop in Step~1 but modifies $Y$ and possibly $R$, 
we distinguish between three cases:

\begin{itemize}
\item[(a)]
The algorithm modifies $Y$ in Step~2(a), by using a $Y$-augmenting arc~$(s,u)$ where $s \notin V$; then $R'=R$.
Note that $s \in $ prefers $Y'$ to~$Y$, and 
for any other agent~$a \in N \setminus R'$ we know $Y(a) =Y'(a)$.
Hence, this modification amounts to deleting all arcs~$(s,a)$ from the envy graph $\GG_{Y \prec}$ 
where $Y(s) \prec_s a \preceq_s Y'(s)$.
\item[(b)]
The algorithm modifies $Y$ in Step~2(b), by using a $Y$-augmenting arc~$(s,u)$ where $s \in V$; then $R'=R$.
First observe that $V \subseteq \QQ$, as the only way the sink set of~$Y$ can change is 
when an agent ceases to be a sink of the current sub-allocation due to the application of Step~2(b). 
Thus, $s \in V$ implies $s \in \QQ$, which means that $(s,u)$ must be the unique arc~$(s,p)$ leaving~$s$. 
Hence, adding $(s,u)$ to~$Y$ amounts to deleting the arc~$(s,u)$ from the envy graph $\GG_{Y \prec}$.
\item[(c)] 
The algorithm modifies $Y$ in Step~3, by adding an agent~$u \in U \setminus V$ 
to the set of irrelevant agents, i.e.,~$R'=R \cup \{u\}$.
Then $Y'(a)=Y(a)$ for each agent~$a \in N \setminus R'$, 
so the envy graph $\GG_{Y' \prec}$ is obtained from $\GG_{Y \prec}$ by deleting~$u$.
\end{itemize}
Since deleting some arcs or a vertex from an acyclic graph results in an 
acyclic graph, the stability of~$Y'$ is clear. 
\qed
\end{proof}

We proceed with the observation that an agent's situation in $Y$ may only improve, unless it becomes irrelevant:
this is a consequence of the fact that the algorithm only deletes arcs and agents from the envy graph $\GG_{Y \prec}$.

\begin{proposition}
\label{obs:happier-and-happier}
Let $Y_1$ and $Y_2$ be two sub-allocations computed by algorithm \algoHM{}, with $Y_1$ computed at an earlier step than $Y_2$, 
and let $a$ be an agent that is not irrelevant at the end of the iteration when $Y_2$ is computed. 
Then either $Y_1(a)=Y_2(a)$ or $a$ prefers $Y_2$ to $Y_1$.
\end{proposition}


In the next two lemmas, we prove that \algoHM{} produces a core allocation.
We start by explaining why irrelevant agents may not become the cause of instability in the housing market.

\begin{lemma}
\label{lem:irrelevant-agents}
At the end of algorithm \algoHM{}, 
there does not exist an arc $(a,b) \in \EE$ 
such that $a \notin R$, $b \in R$ and $Y(a) \prec'_a b$.
\end{lemma}

\begin{proof}
Suppose for contradiction that $(a,b)$ is such an arc, 
and let $Y$ and $R$ be as defined at the end of the last iteration. 
Suppose that \algoHM{} adds $b$ to~$R$ during the $n$-th iteration, and let 
$Y_n$ be the sub-allocation at the beginning of the $n$-th iteration.
By Proposition~\ref{obs:happier-and-happier}, either $Y_n(a)=Y(a)$ or~$Y_n(a) \prec'_a Y(a)$. 
The assumption $Y(a) \prec'_a b$ yields $Y_n(a) \prec'_a b$ by the transitivity of~$\prec'_a$.
Thus, $(a,b)$ is a $Y_n$-augmenting arc entering~$b$, contradicting our assumption that 
the algorithm put $b$ into $R$ in Step~3 of the $n$-th iteration.
\qed
\end{proof}

\begin{lemma}
\label{lem:core}
The output of \algoHM{} is an allocation in the core of~$H'$.
\end{lemma}

\begin{proof}
Let $Y$ and $R$ be the sub-allocation and the set of irrelevant agents, respectively, at the end of algorithm \algoHM{},
and let $U$ be the source set of $Y$.
To begin, we prove it formally that the output $X'$ of~\algoHM{} is an allocation for~$H'$. 

Since \algoHM{} stops only when $U=V$, the arc set $Y$ forms a collection of mutually vertex-disjoint cycles in~$\HH-R$
that covers each agent in $\NN \setminus R \setminus U$;
agents of~$U$ have neither incoming nor outgoing arcs in~$Y$. 
As no agent outside~$\QQ$ can become a sink of~$Y$, we know~$U=V \subseteq \QQ$.

First, assume $\QQ \setminus R \setminus U=\emptyset$, that is, $\QQ \setminus R =U=V$. 
In this case, $Y$ is the union of cycles covering each agent in~$N \setminus R$ exactly once.
Hence, $Y$ is an allocation in the submarket of~$H'$ restricted to agent set~$N \setminus R$,
i.e., $H'_{N \setminus R}$. 

Second, assume $\QQ \setminus R \setminus U \neq \emptyset$.
In this case, $Y$ is the union of cycles covering each agent in~$\NN \setminus R \setminus V$ exactly once.
Let $\qq$ be an agent in $\QQ \setminus R \setminus V$. 
Since $\qq$ is not a sink of~$Y$, is not irrelevant, and has a unique outgoing arc to~$p$, 
we know $(\qq,p) \in Y$. 
As $Y$ cannot contain two arcs entering~$p$, this proves that $\QQ \setminus R \setminus V =\QQ\setminus R  \setminus U = \{ \qq \}$. 
Moreover, since the unique arc entering $\qq$ is from~$q$, we get $(q,\qq) \in Y$.
Therefore, the arc set~$Y \setminus \{(q,\qq),(\qq,p)\} \cup \{(q,p)\}$ is an allocation in $H'_{N \setminus R}$.

Consequently, as $X_R$ is an allocation on $H'_{R \cap N}$, we obtain that $X'$ is indeed an allocation in~$H'$ in both cases.

Now we prove that $X'$ is in the core of $H'$ by showing that the envy graph $G^{H'}_{X' \prec}$ of $X'$ is acyclic.
First, the subgraph $G^{H'}_{X' \prec}[R]$ is exactly the envy graph of $X_R$ in $H'_{R \cap N}$ and hence is acyclic.

\begin{cclaim}
\label{claim:augmenting}
Let $a \in N \setminus R$ and let $(a,b)$ be an $X'$-augmenting arc in $H'$. Then $(a,b)$ is $Y$-augmenting as well, i.e., $Y(a) \prec'_a b$.
\end{cclaim}

\begin{proof}[of Claim]
Suppose first that $(a,b) \notin \{(q,p):q \in Q\}$: then $(a,b)$ is an arc in $G^{\HH}$.
If $a \notin Q$ or~$Y(a) \notin \QQ$, then $Y(a)=X'(a)$  and thus the claim 
follows immediately. If $a \in Q$ and $Y(a)=\widetilde{a} \in \QQ$, then
$X'(a)=p \prec'_a b$ implies that $a$ prefers $b$ to $Y(a)=\widetilde{a}$ in $\HH$ as well, that is, $(a,b)$ is $Y$-augmenting.

Suppose now that $(a,b)=(q,p)$ for some $q \in Q$. We finish the proof of the claim by showing that $(q,p)$ 
is not $X'$-augmenting if $q \notin R$ (recall that we assumed $q=a \notin R$). 

First, if $\qq \notin U$, then necessarily $\{(q,\qq),(\qq,p)\} \subseteq Y$, and so~$(q,p) \in X'$, 
which means that $(q,p)$ is not $X'$-augmenting.

Second, if $\qq \in U$, then consider the iteration in which $\qq$ became a source for our sub-allocation, and 
let~$Y_n$ denote the sub-allocation at the end of this iteration.
Agent~$\qq$ can become a source either in Step~2(a) or in Step~3, 
since Step~2(b) always results in one agent being deleted from the source set without a replacement.
Recall that the only arc entering~$\qq$ is~$(q,\qq)$. 
If $\qq$ became the source of~$Y_n$ in Step~2(a), then we know $\qq \prec'_q Y_n(q)$. 
By Proposition~\ref{obs:happier-and-happier}, this implies $\qq \prec'_q Y(q)$.
By the construction of~$\HH$, we obtain that $q$ prefers $Y(q)=X'(q)$ to $p$ in $H'$, 
so $(q,p)$ is not $X'$-augmenting.
Finally, if agent~$\qq$ became the source of~$Y_n$ in Step~3, then this implies $q \in R$, which contradicts our assumption $a=q \notin R$.
\\ \phantom{a}\hfill $\blacksquare$
\end{proof}

Our claim implies that $G^{H'}_{X' \prec}[N \setminus R]$ is a subgraph of~$\GG_{Y \prec}$ and therefore it is acyclic
by Lemma~\ref{lem:stable-pre-allocation}. 
Hence, any cycle in~$G^{H'}_{X' \prec}$ must contain agents both in~$R$ and in~$N\setminus R$ (recall that $G^{H'}_{X' \prec}[R]$ is acyclic as well).
However, $G^{H'}_{X' \prec}$ contains no arcs from~$N \setminus R$ to~$R$, since such arcs cannot be $Y$-augmenting by Lemma~\ref{lem:irrelevant-agents}.
Thus $G^{H'}_{X' \prec}$ is acyclic and $X'$ is in the core of~$H'$.
\qed
\end{proof}

The following lemma, the last one necessary to prove Theorem~\ref{thm:core-main-RI}, 
shows that \algoHM{} runs in linear time; the proof
relies on the fact that in each iteration but the last either an agent or an arc is deleted from the envy graph, 
thus limiting the number of iterations by $|E|+|N|$.

\begin{lemma}
\label{lem:runtime}
Algorithm \algoHM{} runs in $O(|H|)$ time.
\end{lemma}
\begin{proof}
Observe that the initialization takes $O(|E|+|N|)=O(|E|)$ time; note that $E$ contains every loop $(a,a)$ where $a \in N$, so we have $|E|\geq |N|$. 
We can maintain the envy graph $\GG_{Y \prec}$ in a way that deleting an arc from it when it ceases to be $Y$-augmenting can be done in $O(1)$ time, 
and detecting whether a given agent is entered by a $Y$-augmenting arc also takes $O(1)$ time.
Observe that there can be at most $|E|+|N|$ iterations, since at each step but the last, either an agent 
or an arc is deleted from the envy graph. Thus, the whole iteration takes $O(|E|)$ time. 
Finally, the allocation~$X_R$ for irrelevant agents by the variant of TTC described in Appendix~\ref{sec:app-ttc} can be computed in $O(|H|)$ time. 
Hence, the overall running time of our algorithm is $O(|H|)+O(|E|)=O(|H|)$.
\qed
\end{proof}

We are now ready to prove Theorem~\ref{thm:core-main-RI}.

\begin{proof}[of Theorem~\ref{thm:core-main-RI}]
Lemma~\ref{lem:runtime} shows that algorithm \algoHM{} runs in linear time, 
and by Lemma~\ref{lem:core} its output is an allocation~$X'$ in the core of~$H'$.
It remains to prove that either $X'(p)=X(p)$ or $p$ prefers~$X'$ to~$X$.
Observe that it suffices to show $p \notin R$, by Proposition~\ref{obs:happier-and-happier}.

For the sake of contradiction, assume that \algoHM{} puts~$p$ into the set of irrelevant vertices at some point, 
during an execution of Step~3.
Let $Y$ denote the sub-allocation at the beginning of this step, and let $V$ be its sink set. Clearly, $V \not = \emptyset$ 
(as in that case the source and the sink set of~$Y$ would coincide). Recall also that $V \subseteq \QQ$. 
Thus, there exists some $\qq \in V \subseteq \QQ$. However, then $(\qq,p)$ is an $Y$-augmenting arc by definition, 
entering~$p$, which contradicts our assumption that the algorithm put~$p$ into the set of irrelevant agents in Step~3 of this iteration.
\end{proof}

\subsection{Strict improvement.}
\label{sec:strict-imp}

Looking at Theorem~\ref{thm:core-main-RI} and Corollary~\ref{cor:possible-houses}, 
one may wonder whether it is possible to detect efficiently when a $p$-improvement leads to 
a situation that is strictly better for $p$. 
For a solution concept~$\Phi$ and housing markets $H$ and $H'$ such that $H'$ is a $p$-improvement of $H$ for some agent~$p$, 
one may ask the following questions:

\medskip
\begin{enumerate}
\item \textsc{Possible Strict Improvement for Best House} or \textsc{PSIB}: \\
is it true that $a \prec_p a'$ for some $a \in \Phi(H)^+_p$ and $a' \in \Phi(H')^+_p$?
\item \textsc{Necessary Strict Improvement for Best House} or \textsc{NSIB}: \\
is it true that $a \prec_p a'$ for every $a \in \Phi(H)^+_p$ and $a' \in \Phi(H')^+_p$?
\item \textsc{Possible Strict Improvement for Worst House} or \textsc{PSIW}: \\
is it true that $a \prec_p a'$ for some $a \in \Phi(H)^-_p$ and $a' \in \Phi(H')^-_p$?
\item \textsc{Necessary Strict Improvement for Worst House} or \textsc{NSIW}: \\
is it true that $a \prec_p a'$ for every $a \in \Phi(H)^-_p$ and $a' \in \Phi(H')^-_p$?
\end{enumerate}

\medskip
Focusing on the core of housing markets, it turns out that all of the above four problems are 
computationally intractable, even in the case of strict preferences.

\begin{theorem}
\label{thm:strict-imp}
With respect to the core of housing markets, 
PSIB and NSIB are $\mathsf{NP}$-hard, 
while PSIW and NSIW are $\mathsf{coNP}$-hard, 
even if agents' preferences are strict orders.
\end{theorem}

\begin{proof}
Since agents' preferences are strict orders, we get that 
PSIB and NSIB are equivalent, and similarly, PSIW and NSIW are equivalent as well, since there is a unique best and a unique worst house that an agent may obtain in a core allocation.
Therefore, we are going to present two reductions, one for PSIB and NSIB, and one for PSIW and NSIW. 
Since both reductions will be based on those presented in the proof of Theorem~\ref{thm:arc-in-core}, 
we are going to re-use the notation defined there. 

The reduction for PSIB (and NSIB) is obtained by slightly modifying the reduction from \acycpart{} to \ec{}
which, given a directed graph~$D$ constructs the housing market~$H$.
We define a housing market~$\widehat{H}$ by simply deleting the arc $(b^\star,a^\star)$ from the acceptability graph of $H$. 
Then $H$ is an $a^\star$-improvement of $\widehat{H}$. 
Clearly, as the house of $a^\star$ is not acceptable to any other agent in~$\widehat{H}$,
the best house that $a^\star$ can obtain in any allocation in the core of $\widehat{H}$ is her own. 
Moreover, the best house that $a^\star$ 
can obtain in any allocation in the core of $H$ is either the house of~$b^\star$ or her own. 
This immediately implies that $(\widehat{H},H)$ is a yes-instance of PSIB (and of NSIB) with respect to the core 
if and only if there exists an allocation in the core of $H$ that contains the arc $(a^\star,b^\star)$. 
Therefore, $(\widehat{H},H)$ is a yes-instance of PSIB and of NSIB with respect to the core 
if and only if $D$ is a yes-instance of \acycpart, finishing our proof for PSIB (and NSIB).

The reduction for PSIW (and NSIW) is obtained analogously, 
by slightly modifying the reduction from \acycpart{} to \fec{}
which, given a directed graph~$D$ constructs the housing market~$H'$.
We define a housing market~$\widehat{H}'$ by deleting the arc $(a^\star,s^\star)$ from the acceptability graph of $H'$. 
Then $H'$ is an $s^\star$-improvement of $\widehat{H}'$. 
Clearly, as the house of $s^\star$ is not acceptable to any other agent in~$\widehat{H}'$, 
the worst house that $s^\star$ can obtain in any allocation in the core of $\widehat{H}'$ is her own. 
Moreover, the worst house that $s^\star$ 
can obtain in any allocation in the core of $H'$ is either the house of~$a^\star$ or her own. 
Therefore, $(\widehat{H}',H')$ is a no-instance of PSIW (and of NSIW) with respect to the core 
if and only if there exists an allocation in the core of~$H'$ where $s^\star$ is not trading, i.e., 
that does not contain the arc $(a^\star,s^\star)$. 
So $(\widehat{H}',H')$ is a no-instance of PSIW and of NSIW with respect to the core 
if and only if $D$ is a yes-instance of \acycpart, finishing our proof for PSIW (and NSIW).
\qed
\end{proof}

\section{The effect of improvements in \textsc{Stable Roommates}.}
\label{sec:improvement-SR}

In the \textsc{Stable Roommates} problem we are given a set $N$ of agents, and a preference relation~$\prec_a$ over $N$ for each agent $a \in N$;
the task is to find a stable matching~$M$ between the agents. 
A matching is \emph{stable} if it admits no \emph{blocking pair}, that is, a pair of agents 
such that each of them is either unmatched, or prefers the other over her partner in the matching.
Notice that an input instance for \textsc{Stable Roommates} is in fact a housing market.
Viewed from this perspective, a stable matching in a housing market can be thought of as an allocation that 
(i) contains only cycles of length at most~2, and 
(ii) does not admit a blocking cycle of length at most~2.

For an instance of \SR{}, we assume mutual acceptability, that is, 
for any two agents $a$ and $b$, we assume that $a \prec_a b$ holds if and only if $b \prec_b a$ holds. 
Consequently, it will be more convenient to define the acceptability graph $G^H$ of an instance~$H$ of \SR{}
as an undirected simple graph where agents $a$ and $b$ are connected by an edge $\{a,b\}$ 
if and only if they are acceptable to each other and $a \neq b$. 
A \emph{matching} in $H$ is then 
a set of edges in $G^H$ such that no two of them share an endpoint.


Bir\'o et al.~\cite{BKKV-mor} have shown the following statements, illustrated in Examples~\ref{ex:SRwithTies-violating-RI} and~\ref{ex:SR-violating-RI-worst}.
\begin{proposition}[\cite{BKKV-mor}]
\label{obs:SRwithTies-violating-RI}
Stable matchings in the \SR{} model \begin{itemize}
\item 
violate the RI-worst property (even if agents' preferences are strict), and 
\item violate the RI-best property, if agents' preferences may include ties.
\end{itemize}
\end{proposition}

\begin{example}
\label{ex:SRwithTies-violating-RI}
Let $N=\{a,b,c,d,e,p,q\}$ be the set of agents. The preferences indicated in Figure~\ref{fig:example-SR-ties}
define two housing markets $H$ and $H'$ such that $H'$ is a $(p,q)$-improvement of $H$. 
Note that agent~$d$ is indifferent between her two possible partners.
Looking at $H$ and $H'$ in the context of
\textsc{Stable Roommates}, 
it is easy to see that the best partner that $p$ might obtain in a stable matching for $H$ is her second choice $b$, while 
in $H'$ the only stable matching assigns $a$ to $p$, which is her third choice. 
\end{example}

\begin{figure}[ht]
\begin{center}
\includegraphics[scale=1]{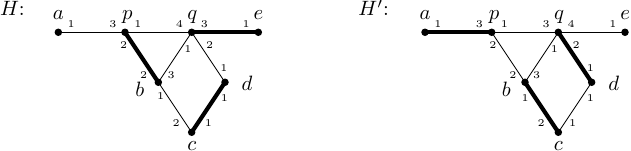}
\caption{The housing markets $H$ and $H'$ as instances of \textsc{Stable Roommates} with ties, in the Example~\ref{ex:SRwithTies-violating-RI}. 
For both $H$ and $H'$, the matching represented by bold  arcs yields the best possible partner for $p$ in any stable matching of the given market.}
\label{fig:example-SR-ties}
\end{center}
\end{figure}

\begin{example}
\label{ex:SR-violating-RI-worst}
Let $N=\{a,b,p,q\}$ be the set of agents. The preferences indicated in Figure~\ref{fig:example-SR-worst}
define two housing markets $H$ and $H'$ such that $H'$ is a $(p,q)$-improvement of $H$. 
The worst partner that $p$ might obtain in a stable matching for $H$ is her top choice $a$, while 
in $H'$ there exists a stable matching that assigns $b$ to $p$, which is her second choice. 
\end{example}

\begin{figure}[t]
\begin{center}
\includegraphics[scale=1]{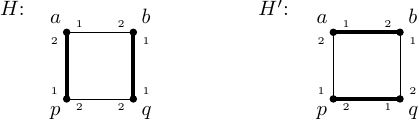}
\caption{The housing markets $H$ and $H'$ in Example~\ref{ex:SR-violating-RI-worst}. 
For both $H$ and $H'$, the matching represented by bold arcs yields the worst possible partner for $p$ in any stable matching of the given market.}
\label{fig:example-SR-worst}
\end{center}
\end{figure}

Complementing Proposition~\ref{obs:SRwithTies-violating-RI}, 
we show that 
a $(p,q)$-improvement can lead to an instance where no stable matching exists at all. 
This may happen even if preferences are strict orders;
hence, stable matchings do not strongly satisfy the RI-best property.

\begin{proposition}
\label{prop:SM-strong-RI-best-fails}
Stable matchings in the \SR{} model do not strongly satisfy the RI-best property, even if agents' preferences are strict.
\end{proposition}

\begin{proof}
Let $N=\{a,b,p,q\}$ be the set of agents. The preferences indicated in Figure~\ref{fig:example-SR-instability}
define housing markets $H$ and $H'$ where $H'$ is an $(p,q)$-improvement of $H$. 
The best partner that $p$ might obtain in a stable matching for $H$ is her second choice $a$, 
while $H'$ does not admit any stable matchings at all.
\qed
\end{proof}

\begin{figure}[t]
\begin{center}
\includegraphics[scale=1]{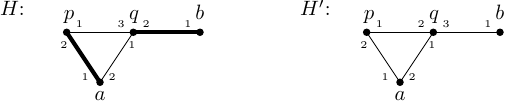}
\caption{Housing markets $H$ and $H'$ illustrating 
the proof of 
Proposition~\ref{prop:SM-strong-RI-best-fails}. 
For $H$, the bold arcs represent a stable matching, while
the instance $H'$, which is a $(p,q)$-improvement of $H$, does not admit any stable matchings.}
\label{fig:example-SR-instability}
\end{center}
\end{figure}

Contrasting Propositions~\ref{obs:SRwithTies-violating-RI} and~\ref{prop:SM-strong-RI-best-fails},
it is somewhat surprising that if agents' preferences are strict, then the RI-best property holds 
for the \textsc{Stable Roommates} setting. Thus, 
the situation of $p$ cannot deteriorate as a consequence of a $p$-improvement unless instability arises.

\begin{theorem}
\label{thm:SR-RI-holds}
Let $H=(N,\{ \prec_a \}_{a \in N})$ be a housing market where agents' preferences are strict orders.
Given a stable matching $M$ in $H$ and a $(p,q)$-improvement $H'$ of $H$ for two agents $p,q \in N$,  
either $H'$ admits no stable matchings at all, or there exists a stable matching $M'$ in $H'$
such that $M(p) \preceq_{p} M'(p)$.
Moreover, given $H$, $H'$ and $M$ it is possible to find such a matching $M'$ in polynomial time, or conclude correctly that $H'$ admits no stable matchings.
\end{theorem}

\begin{corollary}
\label{cor:SR-RI-best}
In the \SR{} model with strict preferences stable matchings satisfy the RI-best property.
\end{corollary}

Comparing our results for the core of housing markets and for stable matchings in the \SR{} model, we find that these solution concepts exhibit  similarities as well as disparities in connection to the notion of respecting improvement. First, neither  solution concept satisfies the RI-worst property (see Propositions~\ref{prop:core-worst-RI-fails} and~\ref{obs:SRwithTies-violating-RI}).
Second, both satisfy the RI-best property, as established by our main algorithmic results, Theorems~\ref{thm:core-main-RI} and~\ref{thm:SR-RI-holds}; however, the core of housing markets \emph{strongly} 
satisfies the RI-best property even if preferences are partial orders, as opposed to stable matchings in the \SR{} model, where it may happen that an improvement yields a market without any stable matchings (see Proposition~\ref{prop:SM-strong-RI-best-fails}), and the RI-best property holds only if preferences are strict, but fails if preferences can contain ties (see Proposition~\ref{obs:SRwithTies-violating-RI}).
Table~\ref{tab:summary-RI} in Section~\ref{sec:future} offers a comparison of these two models from the viewpoint of the property of respecting improvement.

\medskip

We describe our algorithm for Theorem~\ref{thm:SR-RI-holds} in Section~\ref{sec:SR-algo}, 
and prove its correctness in Section~\ref{sec:algoSR-correctness}.

\subsection{Description of algorithm \algoSR{}.}
\label{sec:SR-algo}
To prove Theorem~\ref{thm:SR-RI-holds} we are going to rely on the concept of proposal-rejection alternating sequences
introduced by Tan and Hsueh~\cite{Tan-Hsueh-1995}, originally used as a tool 
for finding a stable partition in an incremental fashion by adding agents one-by-one to a \SR{} instance.
We somewhat tailor their definition to fit our current purposes.

Let $\alpha_0 \in N$ be an agent in a housing market $H$, and let $M_0$ be a stable matching in~\hbox{$H-\alpha_0$}. 
A sequence $S$ of agents $\alpha_0,\beta_1, \alpha_1, \dots, \beta_k,\alpha_k$
is a \emph{proposal-rejection alternating sequence} starting from~$M_0$, 
if there exists a sequence of matchings $M_1, \dots, M_k$ such that for each $i \in \{1, \dots, k\}$
\begin{itemize}[leftmargin=22pt]
\item[(i)] $\beta_i$ is the agent most preferred by $\alpha_{i-1}$ among those who prefer $\alpha_{i-1}$ to their partner in $M_{i-1}$
or are unmatched in $M_{i-1}$,
\item[(ii)] $\alpha_i=M_{i-1}(\beta_i)$, and
\item[(iii)] $M_i=M_{i-1} \setminus \{\{\alpha_i,\beta_i\}\} \cup \{\{ \alpha_{i-1},\beta_i\}\}$ is a matching in $H-\alpha_i$.
\end{itemize}
We say that the sequence~$S$ \emph{starts} from~$M_0$, and that the matchings $M_1, \dots, M_k$ are \emph{induced} by~$S$.
We say that $S$ \emph{stops} at~$\alpha_k$, if there does not exist an agent fulfilling condition~(i) in the above definition for $i=k+1$, 
that is, if no agent prefers $\alpha_k$ to her current partner in~$M_k$ and no unmatched agent in~$M_k$ finds $\alpha_k$ acceptable. 
We will also allow a proposal-rejection alternating sequence to take the form  $\alpha_0,\beta_1, \alpha_1, \dots, \beta_k$, 
in case conditions~(i),~(ii), and~(iii) hold for each $i \in \{1,\dots, k-1\}$, 
and $\beta_k$ is an unmatched agent in~$M_{k-1}$ satisfying condition~(i) for $i=k$.
In this case we define the last matching induced by the sequence as $M_k=M_{k-1} \cup \{\{ \alpha_{k-1},\beta_k\}\}$, 
and we say that the sequence \emph{stops} at agent~$\beta_k$.

We summarize the most important properties of proposal-rejection alternating sequences in Lemma~\ref{lem:pr-seq-prop} 
as observed and used by Tan and Hsueh. Since the first claim of Lemma~\ref{lem:pr-seq-prop} is only implicit in the paper by Tan and Hsueh~\cite{Tan-Hsueh-1995}, we prove it for the sake of completeness. 

\begin{lemma}[\cite{Tan-Hsueh-1995}]
\label{lem:pr-seq-prop}
Let $\alpha_0,\beta_1, \alpha_1, \dots, \beta_k(,\alpha_k)$ be a proposal-rejection alternating sequence
starting from a stable matching $M_0$ and inducing the matchings $M_1, \dots, M_k$ in a housing market $H$.
Then the following hold.
\vspace{-6pt}
\begin{enumerate}
\item $M_i$ is a stable matching in $H-\alpha_i$ for each $i \in \{1, \dots, k-1(, k)\}$.
\item If $\beta_j=\alpha_i$ for some $i$ and $j$, then $H$ does not admit a stable matching; in such a case we say that sequence~$S$ \emph{has a return}. 
\item If the sequence stops at $\alpha_k$ or $\beta_k$, then $M_k$ is a stable matching in $H$. 
\item For any $i \in \{1,\dots, k-1\}$ agent $\alpha_i$ prefers $M_{i-1}(\alpha_i)$ to $M_{i+1}(\alpha_i)$.
\item For any $i \in \{1,\dots, k-1\}$ agent $\beta_i$ prefers $M_i(\beta_i)$ to $M_{i-1}(\beta_i)$.
\end{enumerate}
\end{lemma}

\begin{proof}[of the first statement of Lemma~\ref{lem:pr-seq-prop}]
We prove the statement by induction on $i$; the case $i=0$ is clear. 
Assume that $i \geq 1$ and $M_{i-1}$ is stable in $H-\alpha_{i-1}$. 
Since $M_i \triangle M_{i-1}=\{ \{\alpha_i,\beta_i\},\{ \alpha_{i-1},\beta_i\}\}$
we know that any blocking pair for $M_i$ in $H-\alpha_i$ must contain either $\beta_i$ or $\alpha_{i-1}$. 
By our choice of~$\beta_i$, it is clear that $\alpha_{i-1}$ cannot be contained in a blocking pair. 
Moreover, since $\beta_i$ prefers $\alpha_{i-1}$ to~$M_{i-1}(\beta_i)=\alpha_{i}$, 
any blocking pair for $M_i$ would also be blocking in $M_{i-1}$, a contradiction.
\qed
\end{proof}

We are now ready to describe algorithm \algoSR{}; see Algorithm~\ref{alg:sr-improve} for its pseudocode.

\paragraph{Algorithm \algoSR{}.} 

Let $H=(N,\{\prec_a\}_{a \in N})$ be a housing market containing a stable matching~$M$,
and let $H'=(N,\{\prec'_a\}_{a \in N})$ be a $(p,q)$-improvement of $H$ for two agents~$p$ and~$q$ in~$N$; 
recall that $\prec'_a=\prec_a$ unless $a=q$. 
We now propose algorithm \algoSR{} that computes a stable matching~$M'$ in~$H'$ with $M(p) \preceq_p M'(p)$, whenever $H'$ admits some stable matching.

First, \algoSR{} checks whether $M$ is stable in $H'$, and if so, returns the matching~$M'=M$.
Otherwise, $\{p,q\}$ must be  a blocking pair for~$M$ in~$H'$. 

Second, the algorithm checks whether $H'$ admits a stable matching, and if so, 
computes \emph{any} stable matching $M^\star$ in $H'$ using Irving's algorithm~\cite{Irving-SR}; 
if no stable matching exists for $H'$, algorithm \algoSR{} stops.
Now, if $M(p) \preceq'_p M^\star(p)$, then \algoSR{} returns $M'=M^\star$, otherwise proceeds as follows.

Let $\HH$ be the housing market obtained from $H'$ by deleting all agents $\{a \in N: a \preceq'_q p\}$ 
from the preference list of $q$ (and vice versa, deleting $q$ from the preference list of these agents). 
Notice that in particular this includes the deletion of $p$ as well as of $M(q)$ from the preference list of $q$ (recall that $M(q) \prec'_q p$). 

Let us define $\alpha_0=M(q)$ and $M_0=M \setminus \{q,\alpha_0\}$. 
Notice that $M_0$ is a stable matching in $\HH-\alpha_0$:
clearly, any possible blocking pair must contain $q$, but any blocking pair $\{q,a\}$ that is blocking in~$\HH$ 
would also block $H$ by $M(q) \prec_q a$. 
Observe also that $q$ is unmatched in $M_0$.

Finally, algorithm \algoSR{} builds a proposal-rejection alternating sequence~$S$ consisting of 
agents $\alpha_0, \beta_1, \alpha_1, \dots, \beta_k(, \alpha_k)$ in $\HH$ starting from  $M_0$, and inducing matchings $M_1, \dots, M_k$
until one of the following cases occurs: 
\begin{itemize}
\item[(a)] $\alpha_k=p$: 
      in this case \algoSR{} outputs $M'=M_k \cup \{\{p,q\} \}$;
\item[(b)] $S$ stops:  
      in this case \algoSR{} outputs $M'=M_k$.
\end{itemize}

\begin{algorithm}
\caption{Algorithm \algoSR{}}\label{alg:sr-improve}
\hspace*{\algorithmicindent} \textbf{Input:} housing market~$H=(N,\prec)$, its $(p,q)$-improvement~$H'=(N,\prec')$ for two agents~$p$ and~$q$,   \\
\hspace*{\algorithmicindent} \phantom{\textbf{Input:}} and a stable matching~$M$ in~$H$. \\
\hspace*{\algorithmicindent} \textbf{Output:} a stable matching~$M'$ in~$H'$ such that $M(p) \preceq_p M'(p)$ or $M(p)=M'(p)$, if $H'$ admits\\
\hspace*{\algorithmicindent} \phantom{\textbf{Output:}} some stable matching. 
\begin{algorithmic}[1]
\If{$M$ is stable in~$H'$} \textbf{return}~$M$ 
\EndIf
\If{$H'$ admits a stable matching} 
	let $M^\star$ be any stable matching in~$H'$. 
		\Statex \Comment{Use Irving's algorithm~\cite{Irving-SR}} 
\Else{ {\bf return} ``No stable matching exists for~$H'$.''}
\EndIf 
\If{$M(p) \preceq_p M^\star(p)$} {{\bf return} $M' \leteq M^\star$} 
\EndIf 
\State Create housing market $\HH$ by deleting the agents~$\{a \in N: a \preceq'_q p\}$ from $A(q)$ and vice versa.
\State Set $i \leteq 0$, $\alpha_0 \leteq M(q)$, and $M_0 \leteq M \setminus \{\alpha_0,q\}$
\Repeat \Comment{Computing a proposal-rejection sequence~$S$.}
	\State Set $i \leftarrow i+1$.
	\State Set $B_i \leteq \{ b:\alpha_{i-1} \in A(b), b \textrm{ is unmatched in }M_{i-1} \textrm{ or } M_{i-1}(b) \prec_b \alpha_{i-1}\}$.
	\If{$B_i=\emptyset$}  {{\bf return} $M' \leteq M_{i-1}$} \Comment{$S$ stops at $i-1$.}
		\State Set $\beta_i$ as the agent most preferred by $\alpha_{i-1}$ in~$B_i$.
		\If{$\beta_i$ is unmatched in~$M_{i-1}$} {{\bf return} $M' \leteq M_{i-1} \cup \{\{ \alpha_{i-1},\beta_i\}\}$} \Comment{$S$ stops at $i$.}
			\State Set $\alpha_i \leteq M_{i-1}(\beta_i)$ and $M_i \leteq M_{i-1} \cup \{\{ \alpha_{i-1},\beta_i\}\} \setminus \{\{\alpha_i,\beta_i\}\}$.
		\EndIf
	\EndIf
\Until{$ \alpha_i=p$}
{{\bf return} $M' \leteq M_i \cup \{\{ p,q\}\}$} 
\end{algorithmic}
\end{algorithm}

\subsection{Correctness of algorithm \algoSR{}.}
\label{sec:algoSR-correctness}
To show that algorithm \algoSR{} is correct, we first state the following two lemmas.

\begin{lemma}
\label{lem:algoSR-Sreturns}
The sequence $S$ cannot have a return. Furthermore, if $S$ stops, then it stops at $\beta_k$ with~$\beta_k=q$.
\end{lemma}

\begin{proof}
Recall that $M^\star$ is a stable matching in $H'$ with $M^\star(p) \prec_p M(p)$. 
Since $\{p,q\}$ is a blocking pair for $M$ in $H'$, we know $M(p) \prec_p q$, yielding $M^\star(p) \prec_p q$.
By the stability of $M^\star$, 
this implies that $q$ is matched in $M^\star$ and $p \prec'_q M^\star(q)$.
As a consequence, $M^\star$ is a stable matching not only in $H'$ but also in $\HH$, 
since deleting agents less preferred by $q$ than $M^\star(q)$ from $q$'s preference list cannot compromise the stability of $M^\star$.

By the second claim of Lemma~\ref{lem:pr-seq-prop}, we know that if $S$ has a return, then $\HH$ admits no stable matching,
contradicting the existence of $M^\star$. 
Furthermore, since $q$ is matched in $M^\star$, it must be matched in every stable matching of $\HH$, 
by the well-known fact that in  an instance of \SR{} where agents' preferences are strict 
all stable matchings contain exactly the same set of agents~\cite[Theorem~4.5.2]{GusfieldIrving-book}.
Now, if $S$ stops with the last induced matching~$M_k$, then by the third statement of Lemma~\ref{lem:pr-seq-prop} we get that $M_k$ is a stable matching in $\HH$, 
and thus $q$ must be matched in~$M_k$. Clearly, as $q$ is unmatched in~$M_0$, this can only occur if $\beta_k=q$ and $S$ stops at $q$.
\qed
\end{proof}

\begin{lemma}
\label{lem:algoSR-phaseI}
If \algoSR{} outputs a matching $M'$, then $M'$ is stable in~$H'$ and $M(p) \preceq'_p M'(p)$.
\end{lemma}

\begin{proof}
First, assume that the algorithm stops when $\alpha_k=p$. Then by the first statement of Lemma~\ref{lem:pr-seq-prop}, $M_k$ is stable in $\HH-p$. 
Note also that $q$ must be unmatched in $M_k$, as $q$ can only obtain a partner in the sequence of matchings induced by $S$ if $q=\beta_k$, 
which cannot happen when $\alpha_k=p$. Therefore, $M'=M_k \cup \{\{p,q\}\}$ is indeed a matching in $H'$. 

Let us prove that $M'$ is stable in~$H'$.
Since $q$ is unmatched in~$M_k$, and $M_k$ is stable in~$\HH-p$, no agent acceptable for $q$ prefers $q$ to her partner in~$M_k$ or is left unmatched in~$M_q$. 
Hence, $q$ cannot be contained in a blocking pair for $M'$. 
Thus, any blocking pair for $M'$ must contain~$p$. 
Suppose that $\{p,a\}$ blocks $M'$ in~$H'$; then $q \prec'_p a$. 
Since $S$ cannot have a return by Lemma~\ref{lem:algoSR-Sreturns}, 
we know that $p$ is not among the agents 
$\alpha_0, \beta_1, \alpha_1, \dots, \beta_{k-1}$.
Therefore, $M_{k-1}(p)=M_0(p)=M(p)$. 
Recall that $M(p) \prec'_p q$, which implies $M_{k-1}(p) \prec'_p q$.
Since $M_{k-1}(a)=M_k(a)$ (because $a \notin \{\alpha_{k-1},\beta_k,p\}$), 
we get that $\{p,a\}$ must also block $M_{k-1}$ in $\HH-\alpha_{k-1}$, a contradiction. This shows that $M'$ is stable in $H'$. 
By $M(p) \prec'_p q=M'(p)$, the lemma follows in this case.

Second, assume that \algoSR{} outputs $M'=M_k$ after finding that the sequence $S$ stops with~$q$ being matched in $M_k$.  
By the first statement of Lemma~\ref{lem:pr-seq-prop}, we know that $M'$ is stable in $\HH$, and by the definition of $\HH$, we know that $p \prec_q M'(q)$.
Therefore, $M'$ is also stable in $H'$ (as adding agents less preferred by $q$ than $M'(q)$ to $q$'s preference list cannot compromise the stability of $M'$). 
To show that $M(p) \preceq'_p M'(p)$, it suffices to observe that $p=\alpha_i$ is not possible for any $i \in \{1,\dots,k\}$ 
(as in this case $q$ would be unmatched, as argued in the first paragraph of this proof), and hence by the fifth claim of Lemma~\ref{lem:pr-seq-prop}
the partner that $p$ receives in the matchings $M_0, M_1, \dots, M_k$ can only get better for~$p$, and thus $M(p)=M_0(p) \preceq'_p M_k(p)=M'(p)$.
\qed
\end{proof}

We can now piece together the proof of Theorem~\ref{thm:SR-RI-holds}.

\begin{proof}[of Theorem~\ref{thm:SR-RI-holds}]
From the description of \algoSR{} and Lemma~\ref{lem:algoSR-phaseI} it is immediate that 
any output the algorithm produces is correct. It remains to show that it does not fail to produce an output. 
By Lemma~\ref{lem:algoSR-Sreturns} we know that the sequence $S$ built by the algorithm cannot have a return and can only stop at $q$, 
implying that \algoSR{} will eventually produce an output.
Considering the fifth statement of Lemma~\ref{lem:pr-seq-prop}, we also know that the length of $S$ is at most~$2|E|$. 
Thus, the algorithm finishes in $O(|E|)$ time. 
\qed
\end{proof}

\subsection{A note on strongly stable matchings in \SR{}.}

Given an instance of \SR{} where preferences are not strict, strong stability is an alternative notion of stability based on the notion of weakly blocking pairs. Given a matching~$M$ in a housing market $H=(N,\{ \prec_a\}_{a \in N})$, an edge $\{a,b\}$ in the acceptability graph $G^H$ is \emph{weakly blocking}, if (i) $a$ is either unmatched or weakly prefers $b$ to~$M(a)$, and (ii) $b$ is either unmatched or weakly prefers~$a$ to $M(b)$, and (iii) if $a$ and $b$ are both matched in~$M$, then $a$ prefers $b$ to~$M(a)$, or $b$ prefers $a$ to~$M(b)$. 
If there is no weakly blocking pair for~$M$, then $M$ is \emph{strongly stable}.

Note that a strongly stable matching for~$H$ can be thought of as an allocation that (i) contains only cycles of length at most~2, and (ii) does not admit a \emph{weakly blocking} cycle of length at most~2. 
Recall that stable matchings correspond to the concept of core if we restrict allocations to pairwise exchanges;
analogously, strongly stable matchings correspond to the concept of strict core for pairwise exchanges.
Observe also that if agents' preferences are strict orders, then strong stability is equivalent with stability, or in other words, the strict core and the core coincide.
 
 \medskip
In view of Corollary~\ref{cor:SR-RI-best}, it is natural to ask whether the set of strongly stable matchings satisfies the RI-best property in the case when preferences may not be strict. The following statement answers this question in the negative. Interestingly, the result holds even in the \textsc{Stable Marriage} model, the special case of \SR{} where 
 the acceptability graph is bipartite.

\begin{proposition}
\label{prop:SSM-RI-best-fails}
    Strongly stable matchings in the \textsc{Stable Marriage} model do not satisfy the RI-best property, even if agents' preferences are weak orders.
\end{proposition}

\begin{figure}[t]
\begin{center}
\includegraphics[scale=1]{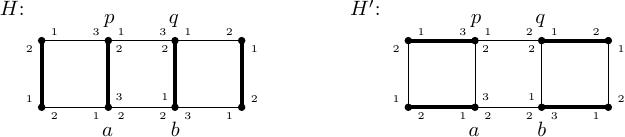}
\caption{The housing markets $H$ and $H'$ in the proof of Proposition~\ref{prop:SSM-RI-best-fails}.
For both $H$ and $H'$, the allocation represented by bold arcs yields the best possible strongly stable matchings. 
}
\label{fig:prop-SSM-RI-best}
\end{center}
\end{figure}

\begin{proof}
Consider the housing markets~$H$ and $H'$ depicted in Figure~\ref{fig:prop-SSM-RI-best}; note that 
$H'$ is a~$(p,q)$-improvement of~$H$. Note that the preferences in~$H$ are strict, but in~$H'$ agent~$q$ is indifferent between~$p$ and~$b$.

First observe that the matching $M$ shown in bold in the first part of Figure~\ref{fig:prop-SSM-RI-best} is stable in~$H$, so it is possible for~$p$ to be matched with its second choice, namely $a$, in a (strongly) stable matching in~$H$.
We claim that the best possible partner $p$ can obtain in any strongly stable matching in~$H'$ is  its third choice. To see this, first note that any matching containing $\{p,q\}$ is weakly blocked by~$\{q,b\}$ in~$H$, 
so~$p$ cannot be matched to its first choice, agent $q$, in any strongly stable matching in~$H'$. 
Second, note that any matching~$M'$ containing $\{p,a\}$ must match $q$ to its first choice (otherwise the pair~$\{p,q\}$ weakly blocks~$M'$) and hence
 $M'$ must match $b$ to its third choice (so as not to form a blocking pair with it); however, then $\{a,b\}$ is a blocking pair for~$M'$. 
Thus, $p$ cannot be matched in any strongly stable matching of~$H'$ to its second choice, agent $a$, either.

By contrast, it is easy to verify that the matching shown in bold in the second part of Figure~\ref{fig:prop-SSM-RI-best}, matching $p$ to its third choice, is
strongly stable in~$H'$.
This proves our proposition.
\qed
\end{proof}

\section{Summary and further research.}
\label{sec:future}

\begin{table}[th]
    \centering
    \begin{tabular}{ll@{\hspace{10pt}}cc@{\hspace{10pt}}cc}
& & \multicolumn{2}{c}{Housing market} & \multicolumn{2}{c}{\SR{}} \\[2pt]
& & Core & Strict core
& Core & Strict core \\
\hline \\[-6pt]
Forced edge/arc:
    &  strict pref. 
    & $\mathsf{NP}$-c  {\small (Thm.~\ref{thm:arc-in-core})}
    & $\mathsf{P}$  {\small \cite{roth-postlewaite}}
    & $\mathsf{P}$  {\small \cite{Fleiner-Irving-Manlove}}
    & $\mathsf{P}$ {\small \cite{Fleiner-Irving-Manlove}}
    \\    
&  weak pref. 
    & $\mathsf{NP}$-c  {\small (Thm.~\ref{thm:arc-in-core})}
    & $\mathsf{P}$ {\small (Thm.~\ref{thm:strictcore-restricted-inP})}
    & $\mathsf{NP}$-c {\small \cite{manlove-hard-variants}}
    & $\mathsf{P}$ {\small \cite{Kunysz16,Kunysz18}}
    \\    
&  partial order pref. 
    & $\mathsf{NP}$-c  {\small (Thm.~\ref{thm:arc-in-core})}
    & open
    & $\mathsf{NP}$-c {\small \cite{manlove-hard-variants}}
    & $\mathsf{NP}$-c {\small \cite{IMS-stacs2003}} 
    \\[8pt]
Forbidden edge/arc:
& strict pref. 
    & $\mathsf{NP}$-c  {\small (Thm.~\ref{thm:arc-in-core})}
    & $\mathsf{P}$  {\small \cite{roth-postlewaite}}
    & $\mathsf{P}$  {\small \cite{Fleiner-Irving-Manlove}}
    & $\mathsf{P}$ {\small \cite{Fleiner-Irving-Manlove}}
    \\    
&  weak pref. 
    & $\mathsf{NP}$-c  {\small (Thm.~\ref{thm:arc-in-core})}
    & $\mathsf{P}$ {\small (Thm.~\ref{thm:strictcore-restricted-inP})}
    & $\mathsf{NP}$-c {\small \cite{Cseh-Heeger-2020}}
    & $\mathsf{P}$ {\small \cite{Kunysz16,Kunysz18}}
    \\    
&  partial order pref. 
    & $\mathsf{NP}$-c  {\small (Thm.~\ref{thm:arc-in-core})}
    & open
    & $\mathsf{NP}$-c {\small \cite{Cseh-Heeger-2020}}
    & $\mathsf{NP}$-c {\small \cite{IMS-stacs2003}} 
    \\[8pt]
Agent trading:
    &  strict pref. 
    & $\mathsf{NP}$-c  {\small (Thm.~\ref{thm:arc-in-core})}
    & $\mathsf{P}$  {\small \cite{roth-postlewaite}}
    & $\mathsf{P}$ {\small \cite{GusfieldIrving-book}}
    & $\mathsf{P}$ {\small \cite{GusfieldIrving-book}}
    \\    
&  weak pref. 
    & $\mathsf{NP}$-c  {\small (Thm.~\ref{thm:arc-in-core})}
    & $\mathsf{P}$ {\small (Thm.~\ref{thm:strictcore-restricted-inP})}
    & $\mathsf{NP}$-c {\small \cite{manlove-hard-variants}}
    & $\mathsf{P}$ {\small \cite{Manlove-TR-99,Kunysz16}}
    \\    
&  partial order pref. 
    & $\mathsf{NP}$-c  {\small (Thm.~\ref{thm:arc-in-core})}
    & open
    & $\mathsf{NP}$-c {\small \cite{manlove-hard-variants}}
    & $\mathsf{NP}$-c {\small \cite{IMS-stacs2003}} 
    \\[6pt]    
    \end{tabular}
    \caption{Summary of known results on the problems of finding an allocation in the core or strict core of a housing market or a \SR{} instance that additionally (i) contains a given arc (or edge), (ii) avoids a given arc (or edge), or (iii) includes a trading cycle containing a given agent. The table classifies each of these problems either as polynomial-time solvable ($\mathsf{P}$) or $\mathsf{NP}$-complete ($\mathsf{NP}$-c for short), except for the cases whose computational complexity remains open. 
    Recall that in an instance of \SR{}, we interpret the core as the set of stable matchings, and the strict core as the set of strongly stable matchings; these two notions coincide for strict preferences.  
    }
    \label{tab:summary-algorithmic}
\end{table}

\begin{table}[th]
    \centering
    \begin{tabular}{ll@{\hspace{10pt}}cc@{\hspace{10pt}}cc}
& & \multicolumn{2}{c}{Housing market} & \multicolumn{2}{c}{\SR{}} \\[2pt]
& & Core & Strict core
& Core & Strict core \\
\hline \\[-6pt]
RI-best:
    &  strict pref. 
    & 
    \checkmark {\small (Cor.~\ref{cor:possible-houses})}
    & \checkmark  {\small \cite{BKKV-mor}}
    & \halfcheckmark 
    {\small (Cor.~\ref{cor:SR-RI-best}, Prop.~\ref{prop:SM-strong-RI-best-fails})}
    & \halfcheckmark 
    {\small (Cor.~\ref{cor:SR-RI-best}, Prop.~\ref{prop:SM-strong-RI-best-fails})}
    \\    
&  weak pref. 
    & \checkmark {\small (Cor.~\ref{cor:possible-houses})}
    & \halfcheckmark {\small \cite{BKKV-mor}}
    & \xmark {\small \cite{BKKV-mor}}
    & \xmark {\small (Prop.~\ref{prop:SSM-RI-best-fails})}
    \\    
&  partial order pref. 
    & \checkmark {\small (Cor.~\ref{cor:possible-houses})}
    & open
    & \xmark {\small \cite{BKKV-mor}}
    & \xmark 
    {\small (Prop.~\ref{prop:SSM-RI-best-fails})}
    \\[8pt]
RI-worst:
& strict pref. 
    & 
    \xmark {\small (Prop.~\ref{prop:core-worst-RI-fails})}
    & \checkmark {\small \cite{BKKV-mor}}
    & \xmark {\small \cite{BKKV-mor}}
    & \xmark {\small \cite{BKKV-mor}}
    \\
& weak pref. 
    & \xmark {\small (Prop.~\ref{prop:core-worst-RI-fails})}
    & \halfcheckmark \cite{BKKV-mor}
    & \xmark {\small \cite{BKKV-mor}}
    & \xmark {\small \cite{BKKV-mor}}
    \\
& partial order pref. 
    & \xmark {\small(Prop.~\ref{prop:core-worst-RI-fails})}
    & open
    & \xmark {\small \cite{BKKV-mor}}
    & \xmark {\small \cite{BKKV-mor}}
\\[6pt]
    \end{tabular}
    \caption[Summary of known results on the property of respecting improvement for the core and the strict core in housing markets and in the \SR{} model.]{Summary of known results on the property of respecting improvement for the core and the strict core in housing markets and in the \SR{} model.
    Symbol \checkmark{} signifies that the given solution concept strongly satisfies the given property (namely, RI-best or RI-worst), 
    while symbol~\halfcheckmark{} 
    means that the given property is satisfied, but not strongly satisfied. Symbol~\xmark{} means that the given property fails to hold.
    }
    \label{tab:summary-RI}
\end{table}

We have investigated questions about the notion of improvement in connection to the core of housing markets. 
Table~\ref{tab:summary-algorithmic} puts into context our algorithmic results on finding  allocations with various restrictions in the core of housing markets; the table presents analogous results for the strict core of housing markets, as well as for the core and strict core of \SR{} instances (when interpreted as the set of stable and strongly stable matchings, respectively).
Table~\ref{tab:summary-RI} summarizes our results regarding the property of respecting improvement both for housing markets and also for the \SR{} model; we also included the known facts about the strict core, mostly by Bir\'o et al.~\cite{BKKV-mor}. 
We remark that several questions remain open in the case when agents' preferences are partial orders; in fact, we are not aware of any result concerning the strict core of housing markets under such preferences.

Even though the property of respecting improvement is deeply connected to agents' incentives in exchange markets,
many solution concepts have not yet been studied from this aspect.
A solution concept that seems interesting from this point of view is the set of stable half-matchings (or equivalently, stable partitions) in instances of \SR{} without a stable matching. 
Although Figure~\ref{fig:example-SR-half} contains an example about stable half-matchings where 
improvement of an agents' house damages her situation, 
perhaps a more careful investigation may shed light on some interesting monotonicity properties.

\begin{figure}[th]
\begin{center}
\includegraphics[scale=1]{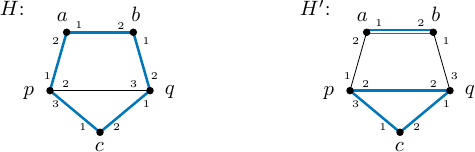}
\caption{An example where an agent's improvement has a detrimental effect on the agent's situation 
in a model where allocations are defined as half-matchings (see also \cite{Tan-1991}). 
Given a \SR{} instance with underlying graph~$(V,E)$, a \emph{half-matching} 
is a function $f:E \rightarrow \{0,\frac{1}{2},1\}$ that satisfies
$\sum_{e=\{u,v\} \in E} f(e) \leq 1$ for each agent~$v \in V$.
The figure contains housing market $H$ and its $(p,q)$-improvement~$H'$, and a unique stable half-matching for each market; 
see~\cite{Manlove2013} for the definition of stable half-matchings.
We depict half-matchings in blue, with double lines for matched edges and single bold lines for half-matched edges.
For~$H$, the half-matching~$f$ depicted leaves $p$ more satisfied than the half-matching~$f'$ depicted for~$H'$.
}
\label{fig:example-SR-half}
\end{center}
\end{figure}

%
%
%


\section*{Acknowledgments.}
Ildik\'o Schlotter is supported by the Hungarian Academy of Sciences under its Momentum Programme (LP2021-2) and its J\'anos Bolyai Research Scholarship.
The research reported in this paper and carried out by Tam\'as Fleiner at the Budapest University of Technology and Economics was supported by the “TKP2020, National Challenges Program” of the National Research Development and
Innovation Office (BME NC TKP2020 and OTKA K143858) and by the Higher Education Excellence Program of the Ministry of Human Capacities in the
frame of the Artificial Intelligence research area of the Budapest University of Technology and Economics (BME FIKP-MI/SC). P\'eter Bir\'o gratefully acknowledges
financial support from the Hungarian Scientific Research Fund, OTKA, Grant No.\ K143858, and the Hungarian Academy of Sciences, Momentum Grant No. LP2021-2.

{\tiny
\bibliographystyle{abbrv}
\bibliography{core}
}


\begin{subappendices}
\renewcommand{\thesection}{\Alph{section}}%

\section{Top Trading Cycles for partial order preferences.}
\label{sec:app-ttc}

Here we present an adaptation of the Top Trading Cycles algorithm for the case when agents' preferences are represented as partial orders; this algorithm always finds an allocation in the core of the given housing market in linear time.
We start by recalling how TTC works for strict preferences, propose a method to deal with partial orders, and finally discuss how 
the obtained algorithm can be implemented in linear time.

\paragraph{Strict preferences.}
If agents' preferences are represented by strict orders, then the TTC algorithm~\cite{shapley-scarf-1974}
produces the unique allocation in the strict core. 
TTC creates a directed graph~$D$ where each agent~$a$ points to her top choice, 
that is, to the agent owning the house most preferred by $a$.
In the graph~$D$ each agent has out-degree exactly~1, since preferences are assumed to be strict. 
Hence, $D$ contains at least one cycle, and moreover, the cycles in~$D$ do not intersect.
TTC selects all cycles in 
$D$ as part of the desired allocation, 
deletes from the market all agents trading along these cycles, and repeats the whole process 
until there are no agents left.  

\paragraph{Preferences as partial orders.}
When preferences are represented by partial orders, 
one can modify the TTC algorithm by letting each agent~$a$ in~$D$ point to her \emph{undominated} choices:
$b$ is undominated for~$a$, if there is no agent~$c$ such that $b \prec_a c$. 
Notice that 
an agent's out-degree is then \emph{at least}~1 in~$D$.
Thus, $D$ contains at least one cycle, but in case it contains more than one cycle, these may overlap. 

A simple approach is to select a set of mutually vertex-disjoint cycles in each round, 
removing the agents trading along them from the market and proceeding with the remainder in the same manner. 
It is not hard to see that this approach yields an algorithm that 
produces an allocation in the core: by the definition of undominated choices, 
any arc of a blocking cycle leaving an agent $a$ necessarily points to an agent that was already removed from the market 
at the time when a cycle containing $a$ got selected. Clearly, no cycle may consist of such ``backward'' arcs only, 
proving that the computed allocation is indeed in the core.

\paragraph{Implementation in linear time.}
Abraham et al.~\cite{ACMM-2004} describe an implementation of the TTC algorithm for strict preferences that runs in $O(|G^H|)$ time.
We extend their ideas to the case when preferences are partial orders as follows. 

For each agent $a \in N$ we assume that $a$'s preferences are given using a \emph{Hasse diagram}
which is a directed acyclic graph~$H_a$ that can be thought of as a compact representation of $\prec_a$.
The  vertex set of~$H_a$ is the set~$A(a)$ of agents whose house is acceptable for~$a$, and it contains an arc~$(b,c)$ if and only if 
$b \prec_a c$ and there is no agent~$c'$ with $b \prec_a c' \prec_a c$. 
Then the description of our housing market~$H$ has length~$\sum_{a \in A}|H_a|$ which we denote by~$|H|$. 
If preferences are weak or strict orders, then $|H|=O(|G^H|)$.

Throughout our variant of TTC, we will maintain a list $U(a)$ containing the undominated choices of $a$ 
among those that still remain in the market, as well as a subgraph~$D$ of $G^H$ spanned by all arcs~$(a,b)$ with $b \in U(a)$.
Furthermore, for each agent~$a$ in the market, we will keep a list of all occurrences of $a$ as someone's undominated choice.
Using $H_a$ we can find the undominated choices of~$a$ in $O(|H_a|)$ time, 
so initialization takes $O(|H|)$ time in total. 

Whenever an agent~$a$ is deleted from the market, we find all agents $b$ such that $a \in U(b)$, 
and we update $U(b)$ by deleting~$a$ and adding those in-neighbors of~$a$ in $H_b$ which have no out-neighbor still present in the market. 
Notice that the total time required for such deletions (and the necessary replacements) to maintain $U(b)$ is $O(|H_b|)$. 
Hence, we can efficiently find the undominated choices of each agent at any point during the algorithm, 
and thus traverse the graph~$D$ consisting of arcs~$(a,b)$ with $b \in U(a)$. 

To find a cycle in~$D$, we simply keep building a path using arcs of~$D$, until we find a cycle (perhaps a loop).
After recording this cycle and deleting its agents from the market (updating the lists $U(a)$ as described above), 
we simply proceed with the last agent on our path. 
Using the data structures described above the total running time of our variant of TTC is $O(|N|+\sum_{a \in N}|H_a|)=O(|H|)$.

\section{Arc restrictions for the strict core.}
\label{sec:app-strictcore}

In this section we investigate the variants of our questions~\Qi{}, \Qii{}, and~\Qiii{} for the strict core of housing markets. Recall that an allocation~$X$ in a housing market~$H=(N,\{\prec_a\}_{a \in N)}$ is in the \emph{strict core} of~$H$ if there is no coalition~$S$ of agents 
with an allocation~$X'$ on~$S$ such that 
(i) $X(a) \preceq_a X'(a)$ for each agent~$a \in S$, and 
(ii) $X(a) \prec_a X'(a)$ for at least one agent~$a \in S$.

Recall that if agents' preferences are strict, then the strict core contains a unique allocation, and this allocation can be efficiently computed by the TTC algorithm~\cite{roth-postlewaite}. Thus, it is trivial to decide whether an agent can obtain a given house in the unique allocation in the strict core.

When agents' preferences are weak orders, then the strict core can be empty~\cite{shapley-scarf-1974}. However, there is a polynomial-time algorithm by Quint and Wako~\cite{WakoQuint} that decides whether the strict core is empty when preferences are weak orders.
We generalize this result by giving a polynomial-time algorithm for the following problem:

\paragraph{The {\sc \arcs{}} problem in housing markets:}
given a housing market~$H$, a set~$F^+$ of \emph{forced} arcs and a set~$F^-$ of \emph{forbidden} arcs in the underlying graph~$G^H$, find an allocation in the strict core of~$H$ that contains~$F^+$ and is disjoint from~$F^-$.

\smallskip
Note that the \arcs{} problem is a generalization of the problems underlying questions~\Qi{}--\Qiii{} when interpreted in relation for the strict core: for a given agent~$a \in N$ and an arc~$(a,b)$ in~$G^H$, 
we can use an algorithm for the \arcs{} problem in order to 
decide whether 
\begin{itemize}
\item     
some allocation in the strict core contains~$(a,b)$, by setting $F^+=\{(a,b)\}$ and $F^-=\emptyset$;
\item some allocation in the strict core avoids~$(a,b)$, by setting $F^+=\emptyset$ and $F^- = \{(a,b)\}$;
\item agent~$a$ is trading in some allocation in the strict core, by setting $F^+=\emptyset$ and $F^-=\{(a,a)\}$.
\end{itemize} 

The following theorem is obtained through a straightforward modification of the algorithm by Quint and Wako for deciding the emptiness of the strict core.
\begin{theorem}
\label{thm:strictcore-restricted-inP}    
If agents' preferences are weak orders, then the \arcs{} problem can be solved in polynomial time.
\end{theorem}

\begin{proof}
Let $H=(N,\{\prec_a \}_{a \in N})$ be the housing market with underlying graph~$G^H=(N,E)$ given as our input, together with arc sets $F^+ \subseteq E$ and~$F^- \subseteq E$. 

We need the following concept from graph theory: we say that a set~$V$ of vertices in a directed graph is an \emph{absorbing set}, if (i) no arc leaves~$V$ and (ii) $V$ is strongly connected, meaning that for each vertices~$v_1,v_2 \in V$ there are paths from~$v_1$ to~$v_2$ and from~$v_2$ to~$v_1$ in the graph.
It is easy to see that two absorbing sets in a directed graph are either identical or vertex-disjoint. 
Recall that an arc~$(a,b) \in E$ is \emph{undominated}, if there exists no agent~$b'$ such that $b \prec_a b'$. Let $U(a)$ denote the set of all undominated arcs in~$G^H$ leaving some agent~$a$, and  let $U=\bigcup_{a \in N} U(a)$. Let $T$ denote the union of all absorbing sets in the subgraph~$(N,U)$ of undominated arcs within~$G^H$, and let~$G_T$ denote the subgraph of~$(N,U)$ induced by vertices of~$T$, i.e., $G_T=(T,\bigcup_{t \in T} U(t))$. Informally speaking, $G_T$ contains the ``top'' agents and their most-preferred choices. 

\medskip
Quint and Wako~\cite{WakoQuint} proved that if agents' preferences are weak orders, then
an allocation~$X$ in~$H$ is in the strict core of~$H$ if and only if
\begin{itemize}
    \item for each agent~$t \in T$, the  arc of~$X$ leaving~$t$ is undominated, and 
    \item  $X[N \setminus T]$ is an allocation in the strict core of $H_{N \setminus T}$.
\end{itemize} 

Using this characterization, it is straightforward to see that the following algorithm solves the \arcs{} problem:

\medskip
\begin{enumerate}
\setlength{\itemindent}{42pt}
\item[{\bf Step~1.\!}] Compute the set~$T$ and the subgraph~$G_T$.
\item[{\bf Step~2.\!}] 
If $F^+$ contains an arc running between~$T$ and~$N \setminus T$, then return `No'.
\item[{\bf Step~3.\!}] 
Find a set~$C \subseteq \bigcup_{t \in T}U(t)$ of arcs in~$G_T$ that is an allocation in the submarket~$H_T$, and (i) contains  all arcs of~$F^+[T]$ and (ii) is disjoint from~$F^-$. 
Return `No' if no such set~$C$ exists.
\item[{\bf Step~4.\!}] 
Use a recursive call to compute an allocation $X'$ in the strict core of the submarket~$H_{N \setminus T}$ that (i) contains all edges of~$F^+[N \setminus T]$, and (ii) is disjoint from~$F^-$.
Return `No' if no such allocation~$X'$ exists.
\item[{\bf Step~5.\!}] Return the allocation~$X=C \cup X'$.
\end{enumerate}

\medskip
From the above characterization of the strict core by Quint and Wako, the correctness of the above algorithm follows immediately. Hence, it remains to check its running time. 

Step~1 can be performed in linear time using standard algorithms on directed graphs, see e.g.,~\cite{KorteVygen-book}. Step~2 takes linear time as well. 
The task in Step~3 can be performed by computing a maximum-weight matching in the following bipartite graph~$\widehat{G}_T$: the vertex set of~$\widehat{G}_T$ is $\{t_1,t_2:t \in T\}$, and for each arc $(t,t')$ in $G_T$ with $(t,t') \notin F^-$ we add an edge $(t_1,t'_2)$ to~$\widehat{G}_T$; the weight of this edge is~$2$ if $(t,t') \in F^+[T]$, otherwise it is~$1$. Then 
allocations\footnote{By an allocation in~$G_T$, we mean an edge set in~$G_T$ that is an allocation in the submarket~$H_T$.} 
in~$G_T$  disjoint from~$F^-$
correspond bijectively to 
perfect matchings in~$\widehat{G}_T$, 
and moreover, 
allocations in~$G_T$ disjoint from~$F^-$ and containing all arcs of~$F^+[T]$
correspond bijectively to
(perfect) matchings in~$\widehat{G}_T$  of weight $|T|+|F^+[T]|$.
Since we can compute a maximum-weight matching in~$\widehat{G}_T$ in~$|T|^3$ time using the Hungarian algorithm (see e.g.,~\cite{KorteVygen-book}), 
the recursion in Step~4 yields an overall running time of~$|N|^3$.
\qed
\end{proof}

Note that Theorem~\ref{thm:strictcore-restricted-inP} is in sharp contrast with our results for the core: by Theorem~\ref{thm:arc-in-core}, 
given a housing market~$H$, 
it is $\mathsf{NP}$-hard to decide whether a given arc of~$G^H$ can be contained in some core allocation of~$H$,  even if agents' preferences are strict orders. 
Hence, in this aspect, we can perceive a computational gap between the core the strict core in housing markets with strictly or weakly ordered preferences.

We remark that we are not aware of any result that would settle the computational complexity of the \arcs{} problem in the case when agents' preferences are partial orders. 
In fact, even deciding the emptiness of the strict core seems to be a problem whose computational complexity is open.

\section{Maximizing the number of agents trading in a core allocation.}
\label{sec:app-max-core-approx}

Perhaps the most natural optimization problem related to the core of housing markets is the following: 
given a housing market $H$, find an allocation in the core of $H$ whose \emph{size}, defined as the number of trading agents, is maximal among all allocations in the core of $H$;
we call this the \textsc{Max Core} problem. 
\textsc{Max Core} is $\mathsf{NP}$-hard by a result of Cechl\'arov\'a and Repisk\'y~\cite{Cechlarova-Repisky-2011}.
In Theorem~\ref{thm:maxcore} below we show that even approximating \textsc{Max Core} is $\mathsf{NP}$-hard. 
Our result is tight in the following sense:
we prove that for any $\varepsilon>0$, approximating \textsc{Max Core} with a ratio of $|N|^{1-\varepsilon}$ 
is $\mathsf{NP}$-hard, where $|N|$ is the number of agents in the market. 
By contrast, a very simple approach yields an approximation with ratio $|N|$.

We note that Bir\'o and Cechl\'arov\'a~\cite{Biro-Cechlarova-2007} proved a similar inapproximability result, 
but since they considered a special model where agents not only care about the house they receive 
but also about the length of their exchange cycle,
their result cannot be translated to our model, and so does not imply Theorem~\ref{thm:maxcore}.
Instead, our reduction relies on ideas we use to prove Theorem~\ref{thm:arc-in-core}.

\begin{theorem}
\label{thm:maxcore}
For any constant $\varepsilon>0$, the \textsc{Max Core} problem is $\mathsf{NP}$-hard to approximate 
within a ratio of $\alpha_{\varepsilon}(N)=|N|^{1-\varepsilon}$ where $N$ is the set of agents, even if agents' preferences are strict orders.
\end{theorem}

\begin{proof}
Let $\varepsilon>0$ be a constant. 
Assume for the sake of contradiction that there exists an approximation algorithm~$\mathcal{A}_\varepsilon$
that given an instance~$H$ of \textsc{Max Core} with agent set $N$ computes in time polynomial in~$|N|$ an allocation 
in the core of~$H$ having size at least $\OPT(H)/ \alpha_{\varepsilon}(N)$, where $\OPT(H)$ is the maximum size of (i.e.,  number of agents trading in) any allocation in the core of $H$.
We can prove our statement by presenting a polynomial-time algorithm for the $\mathsf{NP}$-hard \acycpart{} problem using $\mathcal{A}_\varepsilon$. 

We are going to re-use the reduction presented in the proof of Theorem~\ref{thm:arc-in-core} from \acycpart{} to \ec.
Recall that the input of this reduction is a directed graph $D$ on $n$ vertices, 
and it constructs a housing market~$H$ containing a set $N$ of $4n+4$ agents and a pair~$(a^\star,b^\star)$ of agents
such that the vertices of $D$ can be partitioned into two acyclic sets if and only if 
some allocation in the core of $H$ contains the arc $(a^\star,b^\star)$.
Moreover, such an allocation (if existent) must have size $4n+4$, 
by our arguments in the proof of Theorem~\ref{thm:arc-in-core}.

Let us now define a housing market~$H'=(N',\{\prec'_a\}_{a \in N'})$ 
that can be obtained 
by subdividing the arc~$(a^\star,b^\star)$ with $K$ newly introduced agents $p_1,\dots, p_K$ where 
$$K=\left\lceil (4n+4)^{1/\varepsilon} \right\rceil,$$
that is, we replace the arc~$(a^\star,b^\star)$ with the path~$(a^\star,p_1,p_2,\dots, p_K, b^\star)$; see Figure~\ref{fig:subdivision} for an illustration.
Let $N'=N \cup \{p_1, \dots, p_K\}$.
Formally, we define preferences $\prec'_a$ for each agent $a \in N'$ as follows: 
first, $\prec'_a$ is identical to $\prec_a$ for each $a \in N \setminus \{a^\star\}$;
second, 
$a^\star$ only prefers the house of agent~$p_1$ to her own house; third,
each agent $p_i \in N' \setminus N$ prefers only the house of agent~$p_{i+1}$ to her own house (where we set $p_{K+1}=b^\star$).
Clearly, the allocations in the core of $H$ correspond to the allocations in the core of $H'$ in a bijective manner.
Hence, it is easy to see that if there is an allocation in the core of~$H$ that contains $(a^\star,b^\star)$ and where every agent of $N$ is trading, 
then there is an allocation in the core of~$H'$ where each agent of~$N'$ is trading. 
Conversely, if there is no allocation in the core of~$H$ that contains $(a^\star,b^\star)$, then 
the agents $p_1, \dots, p_K$ cannot be trading in any allocation in the core of $H'$. 
Thus, we have that if $D$ is a yes-instance of \acycpart{}, then $\OPT(H')=|N'|=4n+4+K$; 
otherwise $\OPT(H') \leq 4n+4$.

\begin{figure}[tbh]
\begin{center}
\includegraphics[scale=1]{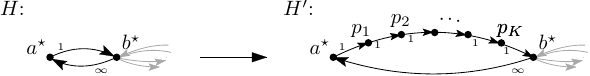}
\caption{Illustration for the proof of Theorem~\ref{thm:maxcore}, for constructing the housing market $H'$ from $H$ by subdividing the arc~$(a^\star,b^\star)$. The figure omits vertices of $N \setminus \{a^\star,b^\star\}$; arcs between $b^\star$ and $N \setminus \{a^\star,b^\star\}$ are shown in gray.
}
\label{fig:subdivision}
\end{center}
\end{figure}

Now, after constructing $H'$ we apply algorithm~$\mathcal{A}_\varepsilon$ with $H'$ as its input; let $X'$ be its output. 
If the size of $X'$ is greater than $4n+4$, then $X'$ must contain at least one vertex from $\{p_1,p_2, \dots,p_K\}$ by $|N|=4n+4$, which 
by the previous paragraph implies that 
there exists an allocation in the core of~$H$ that contains~$(a^\star,b^\star)$, and thus $D$ must be a yes-instance of \acycpart{}.
Otherwise, we conclude that $D$ is a no-instance of \acycpart{}. 
To show that this is correct, it suffices to see that if $D$ is a yes-instance, that is, if $\OPT(H')=|N'|$, then the size of $X'$ is greater than $4n+4$.
And indeed, the definition of~$K$ implies 
$$(4n+4)^{1/\varepsilon}< 4n+4+K = |N'| $$ which raised to the power of $\varepsilon$ yields
$$4n+4< |N'|^{\varepsilon}=\frac{|N'|}{|N'|^{1-\varepsilon}}=\frac{\OPT(H')}{\alpha_{\varepsilon}(N')}$$
as required.

It remains to observe that the above reduction can be computed in polynomial time, because $\varepsilon$ is a constant and so $K$ is a polynomial of $n$ 
of fixed degree.
\qed
\end{proof}

We contrast Theorem~\ref{thm:maxcore} with the observation that an algorithm that outputs \emph{any} allocation in the core yields 
an approximation for \textsc{Max Core} with ratio~$|N|$.

\begin{proposition}
\label{obs:maxcore-trivialapprox}
\textsc{Max Core} can be approximated with a ratio of $|N|$ in polynomial time, where $|N|$ is the number of agents in the input.
\end{proposition}

\begin{proof}
An approximation algorithm for \textsc{Max Core} has ratio~$|N|$, if for any housing market~$H$ with agent set~$N$ 
it outputs an allocation with at least $\OPT(H)/|N|$ agents trading, where $\OPT(H)$ is the maximum number of trading agents in a core allocation of~$H$. 
Thus, it suffices to decide whether $\OPT(H) \geq 1$, and if so produce an allocation in which at least one agent is trading.
Observe that $\OPT(H)=0$ is only possible if $G^H$ is acyclic, as any cycle in $G^H$ blocks the allocation
where each agent gets her own house. Hence, computing \emph{any} allocation in the core of~$H$ is an $|N|$-approximation for \textsc{Max Core};
this can be done in linear time using the variant of the TTC algorithm described in Appendix~\ref{sec:app-ttc}.
\qed
\end{proof}

\end{subappendices}

\end{document}